\documentclass{article}
\usepackage{amssymb} % that is for real numbers etc.
\usepackage{amsmath}
\usepackage{amsfonts}
\usepackage{natbib}
\usepackage{subfigure}
\usepackage{amsthm}   % enables to number all with 1 single counter
\usepackage{epsfig}

\newfont{\handw}{cmmi10 scaled 1200}

\newtheorem{Prop}{Proposition}[section]

\newtheorem{Th}[Prop]{Theorem}

\newtheorem{Def}[Prop]{Definition}

\newtheorem{Cor}[Prop]{Corollary}
 
\newfont{\smcal}{cmu10 scaled 1200}

\newcommand{\Kern}{\operatorname {kern}}
\newcommand{\id}{\operatorname {id}}

\newcommand{\tr}{\operatorname {tr}}
\newcommand{\rank}{\operatorname {rank}}
\newcommand{\lw}{\mbox{\handw \symbol{96}}}

\begin{document}
%      \title{Parallel Shape Evolution and Comparison of Leaf Growth Dynamics}
      \title{Dynamic Shape Analysis and\\ Comparison of Leaf Growth}
%      \title{Parallel Shape Evolution with Application to sLeaf Growth Comparsion}
	%: Comparing Shape Evolution of Leaves over Time} 
   \author{Stephan Huckemann}
%   \date{April 2005}
    \maketitle

% \begin{figure}
%  \includegraphics[angle=-90,width=1\textwidth]{../IMG/parallelTP_be1b2_be2b1.eps}
%  \includegraphics[angle=-90,width=1\textwidth]{../IMG/parallelTP_be1b9_be1b8.eps}
% \end{figure}
% \begin{figure}
%  \includegraphics[angle=-90,width=1\textwidth]{../IMG/parallelTP_be2b1_be1b2.eps}
%  \includegraphics[angle=-90,width=1\textwidth]{../IMG/parallelTP_be2b1_be1b9.eps}
% \end{figure}
% \begin{figure}
%  \includegraphics[angle=-90,width=1\textwidth]{../IMG/parallelTP_be2b1_be1b8.eps}
%  \includegraphics[angle=-90,width=1\textwidth]{../IMG/parallelTP_be1b8_be1b9.eps}
% \end{figure}

\begin{abstract}In the statistical analysis of shape a goal beyond the analysis of static shapes %one-way MANOVA 
lies in the quantification of  `same' deformation of different shapes. Typically, shape spaces are modelled as Riemannian manifolds on which parallel transport along geodesics naturally qualifies as a measure for the `similarity' of deformation. Since these spaces are usually defined as combinations of Riemannian immersions and submersions, only for few well featured spaces such as spheres or complex projective spaces (which are Kendall's spaces for 2D shapes), parallel transport along geodesics can be computed explicitly. %, cf. \cite{HHM09}.
 In this contribution a general numerical method to compute parallel transport along geodesics when no explicit formula is available is provided. This method is applied to the shape spaces of closed 2D contours based on angular direction %introduced by \cite{ZR72} as well as introduced and extended by \cite{KSMJ04} 
 and to Kendall's spaces of shapes of arbitrary dimension. In application to the temporal evolution of leaf shape over a growing period, one leaf's shape-growth dynamics can be applied to another leaf. For a specific poplar tree investigated it is found that leaves of initially and terminally different shape evolve rather parallel, i.e. with comparable dynamics.\end{abstract}

% \nocite{ZR72} \nocite{KSMJ04} \nocite{DM98} \nocite{HHM09}\nocite{HHM07}\nocite{HHGMS07} \nocite{O66} %\nocite{L99}
% \nocite{MTY06} \nocite{SCC06} \nocite{HPV02} \nocite{FS08}
\par
\vspace{9pt}
\par
\vspace{9pt}
\noindent {\it AMS 2000 Subject Classification:} \begin{minipage}[t]{6cm}
Primary 62H35\\ Secondary 53C22
 \end{minipage}
\par

%\begin{document}
%\begin{paper}

\section{Introduction}\label{intro_scn}

	For more than two millenia, the analysis of form of biological entities has been an enticing object of human occupation. While early work tended to be speculative in nature, the current state of mathematics and computational power allow to develop and simultaneously verify theoretical results, thereby increasingly driving  scientific progress as witnessed today.

	This present work has been motivated by joint research with the Institute for Forest Biometry and Informatics at the University of G\"ottingen, to compare leaf growth dynamics within single specimen, species and taxa for identification of gene expression. The endeavor is challenging as it touches problems at least as old as Theophrastus' (371 -- 287 b.C.) famous book on ``plant growth'', cf. \cite{Th76}.
	
	The first step of this project is the subject of this work: to develop a framework allowing to compare shape dynamics. To this end, we model biological growth by (generalized) geodesics in shape space. We do so because the geometry of shape spaces in which travel along geodesics requires no energy seems linked to the physiological reality of growth preferring to minimize energy. This ``geodesic hypothesis'', originally stated by \cite{LK00}) is further supported by earlier research, cf. \cite{HZ06}; \cite{HHGMS07}. As with the ``geodesic hypothesis'' one can say that \emph{geodesic shape deformation} of two different shapes is the ``same'' if the impetus of the first deformation is transplanted to the second with no loss of energy. In the language of Riemannian geometry this translates to the condition that the initial velocity of the second geodesics is the \emph{parallel transport} of the initial velocity of the first geodesic. If the deformations are not the same, i.e. the geodesics are not parallel at the first and the second shape, this  concept gives a correlation-based distance between the deformations.

	%The first step of this project is the subject of this work: to develop a framework allowing to compare shape dynamics. To this end, we model biological growth by (generalized) geodesics in shape space. This ``geodesic hypothesis'' (originally stated by \cite{LK00}) is theoretically plausible since travel along geodesics requires no energy thereby allowing nature to minimize energy. Also, practical research supports this hypothesis, e.g \cite{HZ06} or \cite{HHGMS07}. As with the ``geodesic hypothesis'' one can say that \emph{geodesic shape deformation} of two different shapes is the ``same'' if the impetus of the first deformation is transplanted to the second with no loss of energy. In the language of Riemannian geometry this translates to the condition that the initial velocity of the second geodesics is the \emph{parallel transport} of the initial velocity of the first geodesic. If the deformations are not the same, i.e. the geodesics are not parallel at the first and the second shape, this  concept gives at once a distance between the deformations.

	In consequence, the aim of this paper is to provide for parallel transport on shape spacs. Recall that most shape spaces can be viewed as Riemannian immersions or submersions or, combinations thereof.
	Explicit formulae for parallel transport are only available for special spaces. e.g. for spheres and Kendall's spaces of planar shapes, cf. \cite{HHM09}. In general, parallel transport may be difficult to compute and be only available numerically. In the following Section \ref{Para-scn} we provide for a general method to compute parallel transport on shape spaces. In view of our application the method is illustrated in Section \ref{ZR-scn} for the spaces of closed 2D contours based on angular direction with and without specific initial point (cf. \cite{ZR72} as well as \cite{KSMJ04}), and in Section \ref{Kendall_scn} for Kendall's landmark based shape spaces (e.g. \cite{DM98}). 

	In Section \ref{k-rot-inv:scn}, we compare parallel transport on the spaces of closed contours with parallel transport on Kendall shape spaces for simple regular polygonal configurations. While all sectional curvatures in Kendall's shape space are bounded from below by $1$, it turns out that the corresponding subspace of closed contours is flat.
 
	Finally in Section \ref{App:scn}, leaf growth of one leaf is transported parallelly to other leaves and both shape evolutions are compared with one another. For a specific Canadian black poplar investigated we find that leaves with initially and terminally different shapes tend to evolve parallel, in particular so if no shape anomalies are present. Thus the \emph{geodesic hypothesis} %of \cite{LK00} stating that ``biological growth tends to follow geodesics in shape spaces'' 
	can be extended to the \emph{parallel hypothesis}:
	\begin{center}{\it 
		biological growth of related objects, possibly of initially and terminally different shape, tends to follow parallel geodesics,}
	\end{center}
	Using Euclidean approximations in landmark based shape spaces rather than geodesics, this hypothesis was originally coined by \cite{MKMA00} who observed parallel growth patterns.
	Readers primarily interested in the application can directly skip to Section \ref{App:scn}.

\section{Parallel Transport (PT)}\label{Para-scn}

	This section begins with a review of basic concepts of Riemannian geometry found in any standard textbook (specifically \cite{L99} is very appropriate for the following), in particular formulae relating covariant derivatives of Riemannian immersions and submersions. These provide differential equations lifting the parallel transport on shape space to Euclidean or Hilbert space. 

	For a Riemannian manifold $M$, possibly of countable dimension denote by $\langle V_p,W_p\rangle^M$ the \emph{Riemannian metric}
	of tangent spaces %, by $d^M:M\times M \to [0,\infty)$ the induced distance on $M$ 
	and by $\nabla^M_VW$ the \emph{covariant derivative} of vector-fields. Here $V,W\in T(M)$ denote vector-fields with values $V_p,W_p$ in the tangent space $T_pM$ of $M$ at $p\in M$. $d^M(p,p')$ is the induced metrical distance on $M$ for $p,p'\in M$,  $V\otimes W$ denotes the \emph{outer product} defined by $(V\otimes W)\,X = \langle X,V\rangle\,W$. A vector-field $W\in T(M)$ is \emph{parallel} along a smooth curve $t\to \gamma(t)$ on $M$ if it satisfies the ordinary differential equation (ODE)
	\begin{eqnarray}\label{parallel_eq}
	\nabla_{\dot{\gamma}}W &=& 0\,.
	\end{eqnarray}
% 	Subsequent arguments exploit the fact that the \emph{parallel equation} (\ref{parallel_eq}) written in local coordinates is a (possibly infinite-dimensional) system of first order ordinary differential equations (ODE) for $W$. 
	It is well known that there is locally a unique solution $W$ along $\gamma$ for a given initial value. In Euclidean or Hilbert space the left hand side %of the system 
	has the simple form (\ref{Eucl_nabl_eq}).
	
	In particular, geodesics are characterized by the fact that their velocity is parallel:
	$\nabla_{\dot{\gamma}}\dot{\gamma} = 0\,.$
	
	The covariant derivative is often called a \emph{covariant connection}. Indeed, if two offsets $p,p'\in M$ can be joined by a unique geodesic segment of minimal length, their respective tangent spaces are \emph{connected} via \emph{parallel transport} (PT). 
	\begin{Def}\label{parallel_transp_def}
	$w' \in T_{p'}M$ is the \emph{parallel transplant} %$\theta_{p,p'}(w)$ 
	of $w\in T_pM$ if there are 
	\begin{enumerate}
	 \item 	a unique unit speed geodesic $t\to \gamma(t)$  connecting $p=\gamma(0)$ with $p'=\gamma\big(d^M(p,p')\big)$, and
	\item a vector field $W\in T(M)$ parallel along $\gamma$ with $W_p=w, W_{p'}=w'$.
	\end{enumerate}
	\end{Def}

	A sufficient condition for the existence of such a unique connecting geodesic is that $M$ is finite dimensional and $p'$ is sufficiently close to $p$. In case of infinite dimension, examples of complete spaces can be constructed which do not feature minimizing geodesics between arbitrary close points (e.g. \citet[pp.226/7]{L99}). For our applications in mind this fact seems less troublesome since infinite dimensional spaces considered here are built from projective limits of finite dimensional spaces.

	The Euclidean and Hilbert spaces $\mathbb R^n$ (for Hilbert space $n=\infty$) can be identified with all of their tangent spaces, i.e. $\langle v,w\rangle^{\mathbb R^n} = \sum_{i=1}^n v^iw^i$  and the covariant derivative is just the usual multivariate derivative by components,
	$$\nabla^{\mathbb R^n}_{(v^1,\ldots,v^n)}(w^1,\ldots,w^n) = \sum_{i=1}^nv^i\left(\frac{\partial w^1}{\partial x^i},\ldots, \frac{\partial w^n}{\partial x^i}\right)\,.$$
	In particular, if $v=\dot{x}(t)$, i.e. $\frac{dx^i}{dt} = v^i$ we have that
	\begin{eqnarray}\label{Eucl_nabl_eq}
	\nabla^{\mathbb R^n}_{\dot{x}(t)} W &=& \frac{d~}{dt}\, W_{x(t)}\,.
	\end{eqnarray}

	Thus as desired, parallel transport on Euclidean and Hilbert spaces is given by affine translations. 

	For short we write $W(t)$ for the value of $W$ along a selfunderstood smooth curve $t\to\gamma(t)$ and 
	$$\dot{W}(t) ~:= ~\frac{d~}{dt}\, W_{\gamma(t)}$$
	in the Euclidean/Hilbert case.

%\subsection{Riemannian Immersions and Submersions}
	A surjective linear mapping $f : E\to F$ of topological vector-spaces \emph{splits in $F$} if $\Kern(f)$ has a closed complement $\widetilde{F}$ in $E$ such that $ \widetilde{F}\times\Kern(f)\cong E$ as topological vector-spaces, in particular, $ \Kern(f)\to E\to F$ is a \emph{short exact sequence}. Another wording is that \emph{$F$ splits over $E$}.
 
	A smooth mapping $\Phi : M \to N$ of Riemannian manifolds $M$ and $N$ induces a differential mapping $d\Phi_p:T_pM \to T_{\Phi(p)}N$ of tangent spaces. $\Phi$ is called  
	\begin{enumerate}
	\item[] an \emph{immersion} if $\Phi$ is injective and if every $T_{\Phi(p)}N$ splits over $d\Phi_pT_pM$,
	\item[] a \emph{submersion} if $\Phi$ is surjective and if every $d\Phi_p$ splits in $T_{\Phi(p)}N$,
	\item[] an \emph{isometry} if 
	$\langle V_p,W_p\rangle^M = \langle d\Phi_pV_p,d\Phi_pW_p\rangle^N\,,~~\forall p \in M\mbox{ and }V,W \in T(M)\,.$
	\end{enumerate}
	An isometric immersion (submersion) is a \emph{Riemannian immersion (submersion)} respectively. 
	
	\paragraph{Riemannian Immersions.} If $\Phi:M\to N$ is a Riemannian immersion then the tangent spaces of $N$ split into the tangent spaces of $\Phi(M)$ and its orthogonal complements, the \emph{normal spaces}
	\begin{eqnarray*}%\label{Rie-subm-tangent-space-dec}
	T_{\Phi(p)} N &=& T_{\Phi(p)}\Phi(M) \oplus N_{\Phi(p)}\Phi(M)\,.
	\end{eqnarray*}
	As a consequence of the implicit function theorem, every Riemannian immersion $\Phi:M\to N$ admits locally an implicit representation $\Psi : U\cap N \to  N_{\Phi(p)}\Phi(M)$ such that $\Psi(U\cap M) = U \cap \Phi(N)$. Here $U$ is a suitable neighborhood of $\Phi(p)$ in $N$. Hence, we have with $X,Y\in T(M)$ and arbitrary local extensions $\widetilde{X},\widetilde{Y} \in T(U\cap N)$ of $d\Phi X, d\Phi Y \in TM$  that
	\begin{eqnarray}\label{imm-cov-der-eq}
	(\id_{T(N)}-d\Psi)\left(\nabla^N_{\widetilde{X}} \widetilde{Y}\right) &=& 
	d\Phi \left(\nabla^M_X Y\right)\,. 
	\end{eqnarray}
	In particular, $d\Psi_{\Phi(p)}$ spans the normal space $N_{\Phi(p)}\Phi(M)$. 

	\begin{Th}\label{imm-parallel-th} Suppose that an embedding $id_M:M\hookrightarrow \mathbb R^n$ is a Riemannian immersion in Euclidean ($n<\infty$) or Hilbert space ($n=\infty$), $t\to\gamma(t)$ a geodesic in $M$, $\{V_j(t):j\in J\}$ an orthonormal smooth base for $N_{\gamma(t)}M$ and $W$ a vector-field in $M$. Then $W$ is parallel along $\gamma$ if and only if it satisfies the linear differential equation
	\begin{eqnarray*}%\label{imm-parallel-formula}
	 \dot{W}(t)&=&-\left(\sum_{j\in J} \dot{V}_j(t)\otimes V_j(t)\right)\,W(t)\,.
	\end{eqnarray*}
	\end{Th}

	\begin{proof}
	The assertion is an immediate consequence of (\ref{imm-cov-der-eq}) and the fact that
	$$0~=~\frac{d}{dt} \big\langle W(t),V_j(t)\big\rangle ~=~ \big\langle \dot{W}(t),V_j(t)\big\rangle + \big\langle W(t),\dot{V}_j(t)\big\rangle$$
	for all $j\in J$ by hypothesis.	 
	\end{proof}

	\paragraph{Riemannian Submersions}
	For a Riemannian submersion $\Phi : M \to Q$  from the \emph{top space} $M$ to the \emph{bottom space} $Q$, tangent spaces split as follows:
	every fiber $\Phi^{-1}(q)$, $q\in Q$ is a submanifold of $M$ that is locally a topological embedding. With the \emph{vertical space} $T_p \Phi^{-1}\big(\Phi(p)\big)$ along the fiber and its orthogonal complement, the \emph{horizontal space}, we have 
	$$T_p M =  T_p \Phi^{-1}\big(\Phi(p)\big) \oplus H_pM\,.$$
	Since $H_pM \cong T_{\Phi(p)} Q$, every $V\in T(Q)$ has a unique horizontal lift $\widetilde{V} \in H_pM$ characterized by $d\Phi \widetilde{V} = V$. For arbitrary $W\in T(M)$ denote by $W^\perp: p \to W^\perp_p$ the orthogonal projection to the vertical space. 

	The following Theorem due to \cite{O66} (cf. also \citet[p.386]{L99}) allows to lift bottom space parallel transport to the top space. In addition to (\ref{imm-cov-der-eq}) this provides the vertical (normal) part as well, which is in general non-zero for submersions.

	\begin{Th}\label{ONEILL_thm} Let $\Phi:M\to Q$ be a Riemannian submersion and let $X,Y\in T(Q)$. Then we have with the \emph{Lie bracket} $[\cdot,\cdot]$ on $M$ that
	\begin{eqnarray*}%\label{O_Neil_formula}
	\nabla^M_{\widetilde{X}} \widetilde{Y} &=& 
		\widetilde{\nabla^N_X Y} 
		+ ~\frac{1}{2}~ [\widetilde{X},\widetilde{Y}]^{\perp}\,.
	\end{eqnarray*}
	\end{Th}

	We are now ready for the ODE of parallel transport on a Riemannian immersion followed by a Riemannian submersion. 

	\begin{Th}\label{parallel_riem_subm:thm}
	 Suppose that $\Phi_1: M\hookrightarrow \mathbb R^n$ is a Riemannian immersion in Euclidean ($n<\infty$) or Hilbert space $(n=\infty)$, $\Phi_2:M \to Q$ a Riemannian submersion and let $W$ be a vector field on $M$ horizontal along a horizontal geodesic $\gamma(t)$ on $M$. Then $d\Phi_2 W$ is parallel along $\Phi_2\circ \gamma(t)$ if and only if
	\begin{eqnarray*}%\label{subm-parallel-formula}
	 \dot{W}(t)&=&-\left(\sum_{j\in J} \dot{V}_j(t)\otimes V_j(t)\right)\,W(t) - \sum_{k\in K}d\omega^t_k\big(\dot{\gamma}(t),W(t)\big)\,.
	\end{eqnarray*}
	Here, $\{V_j(t):j\in J\}$ denote an orthonormal smooth base for the normal space $N_{\gamma(t)}M\subset \mathbb R^n$ and $d\omega^t_k$ are the exterior derivatives of an orthonormal and smooth base $\{U_k(t): k \in K\}$ of the vertical space $T_{\gamma(t)}[\gamma(t)] \subset T_{\gamma(t)}M$ for suitable index sets $J$ and $K$.
	\end{Th}

	\begin{proof}
	 Suppose that we have a vector field $X\in T(Q)$ with horizontal lift $\widetilde{X}(t) =\dot{\gamma}(t)$. If $\Phi_2\circ \gamma$ is geodesic and $d\Phi_2 W$ parallel with horizontal lift $W$, Theorem \ref{ONEILL_thm} yields
	\begin{eqnarray}\label{O'N:eq}
	 \nabla_{\dot{\gamma}}^MW&=&\frac{1}{2} \sum_{k\in K}\langle [\widetilde{X},W],U_k\rangle]U_k ~=~-\sum_{k\in K} d\omega_k(\dot{\gamma},W)
	\end{eqnarray}
	making use of the well known (e.g. \citet[p.126/7]{L99})
	$$\langle [\widetilde{X},W],U_k\rangle = \widetilde{X}\langle U_k,W\rangle - W\langle U_k,\widetilde{X}\rangle - 2d\omega_k(\widetilde{X},W)$$
	with the exterior derivative $d\omega_k$ of the one-form $\omega_k$ dual to $U_k$. On the other hand, since $d\Phi_1:T_pM\to T_pM \subset \mathbb R^n$ is the identity, formula (\ref{imm-cov-der-eq}) yields 
	\begin{eqnarray}\label{imm:eq}\nonumber
	 \nabla_{\dot{\gamma}}W^M &=& \dot{W}(t) - \left(\sum_{j\in J} {V}_j(t)\otimes V_j(t)\right)\,\dot{W}(t)\\&=&  \dot{W}(t)+\left(\sum_{j\in J} \dot{V}_j(t)\otimes V_j(t)\right)\,W(t)
	\end{eqnarray}
	as in the proof of Theorem \ref{imm-parallel-th} . Putting together (\ref{O'N:eq}) and (\ref{imm:eq}) gives the assertion of the Theorem.
	\end{proof}

% 	As an immediate consequence of Theorem \ref{ONEILL_thm} and the parallel equation (\ref{parallel_eq}) we have the following Corollary.
% 
% 	\begin{Cor}\label{parallel_transo_cor} Let  $\Phi:M\to N$ be a Riemannian submersion, $t \to \gamma(t)$ a horizontal geodesic on $M$ and $W(t)$ a horizontal vector-field along $\gamma$. Then $d\Phi W$ is parallel along $\delta = \Phi\circ\gamma$ if and only if
% 	\begin{eqnarray}\label{O_Neil_parallel_formula}\nabla^M_{\dot{\gamma}(t)} W(t) &=& \frac{1}{2}\,\big[\dot{\gamma}(t),W(t)\big]^{\perp}\,.
% 	\end{eqnarray}
% 	\end{Cor}
% 
% 	Using the well known (e.g. \citet[p.126/7]{L99})
% 	$$\langle V,[Y,Z]\rangle = Y\langle V,Z\rangle - Z\langle V,Y\rangle - 2d\omega(Y,Z)$$
% 	with the exterior derivative $d\omega$ of the one-form $\omega$ dual to $V$, we have
% 	\begin{eqnarray}\label{single_vert_parallel_transo}
% 	 \frac{1}{2}\,\left\langle V(t),\big[\dot{\gamma}(t),W(t)\big]\right\rangle &=& -d\omega\big(\dot{\gamma}(t),W(t)\big)
% 	\end{eqnarray}
% 	for a horizontal geodesic $\gamma$, vector-fields $W$ horizontal and  $V$ vertical along $\gamma$. 
 
\section{PT for Closed 2D Contours}\label{ZR-scn}

	We define the two \emph{shape spaces of closed 2D constant-speed contours based on angular direction} as introduced by \cite{ZR72} in the geometric formulation of \cite{KSMJ04}. % as well as \cite{SCC06}.

	Suppose that $z:[0,2\pi] \to \mathbb C, s\mapsto z(s)$ % = r(t)e^{i\phi(t)}$ 
	is a constant-speed parameterization of a  smooth, closed, curve of length $L$ winding once counterclockwise around each interior point. Let %\cite{ZR72} defined 
	\begin{eqnarray} \nonumber %\label{ZR-pre-shape-def1}
	\theta(s) &=&\arg\big(z'(s)\big) -\arg\big(z'(0)\big) -s\,,\mbox{ with}\\ \label{ZR-ode}
	\dot{z}(s) &=& %\big(\dot{r}(t) + ir(t)\,\dot{\phi}(t)\big)\,e^{i\phi(t)} ~=~
	\frac{L}{2\pi} \,e^{i\big(\theta(s)+\arg(z'(0)) +s\big)}\,.
	\end{eqnarray}
	Obviously, the \emph{Zahn-Roskies shape} (ZR-shape) $\theta$ is invariant under translation, scaling and rotation $z(s) \to c + \lambda e^{i\psi} z(s)$. Moreover, subtracting $s$ (the curves to be modelled wind once around their interior) norms $\theta$ such that it is $2\pi$-periodic. Vice versa, from every  converging Fourier series 
% 	converging in $\lw^2$
% 	\begin{eqnarray}\label{ZR-pre-shape-def2}
% 	\theta(t) &=& \Re\left(\sum_{n=1}^{\infty} a_n e^{int}\right)
% 	\end{eqnarray}
	an a.e. differentiable constant-speed 2D curve  can be reconstructed by integrating (\ref{ZR-ode}). This curve is unique modulo translation, scaling and rotation. Thus a linear subspace of the Hilbert space $\lw^2$ of Fourier series is the \emph{ZR--pre--shape space}
%	$$\begin{array}{rcl}
	\begin{eqnarray*}
	S_{ZR} ~:=~\Big\{ \theta(s) = \sum_{n=0}^{\infty} \big(x_n\cos(ns)+y_n\sin(ns)\big):\\ %\|\theta\|_{S_{ZR}}^2 := 
	\|\theta\|^2-x_0^2= \frac{1}{2}\,\sum_{n=1}^{\infty}(x_n^2+y_n^2)<\infty\\ x_0 = -\sum_{n=1}^\infty x_n,~y_0=0\Big\}\,. 
	\end{eqnarray*}
%	\end{array}$$ 
	As usual, $2\pi \langle\theta,\eta\rangle := \int_0^{2\pi} \theta(s)\eta(s)\,ds$ and $\|\theta\|^2 := \langle \theta,\theta\rangle$. The tangent spaces $T_{\theta}S_{ZR}$ are identified with $S_{ZR}\subset \lw^2$.  
	Since the curves in question are closed, we have with the non-linear mapping
	\begin{eqnarray*}%\label{ZR-implicit-fcn}
	%\left.
	\begin{array}{rcl}
	\Psi : \lw^2 &\to& \mathbb C\\
		\theta &\mapsto&  \int_0^{2\pi} e^{i\big(\theta(s)+s\big)}\,dt  	
	  \end{array}%\right\}
	\end{eqnarray*}
	that the \emph{ZR--shape space} is the implicit sub-manifold 
	\begin{eqnarray*}%\label{ZR-shape-def}
	\Sigma_{ZR} ~:=~ \left\{ \theta \in S_{ZR}: \Psi(\theta) =0\right\}\,.&&
	\end{eqnarray*}
	Obviously, the ZR-shapes of closed not self-intersecting contours form an open subset containing the origin, which corresponds to the shape of the circle.
	%\cite{ZR72}  
% 	\begin{enumerate}
% 	 \item[(a)] expand $\Sigma_{ZR}$ to piece-wise continuous $\theta$ in order to modell piece-wise smooth shapes such as polygons, and more gravely they 
	%note that
% 	\item[(b)] 
	%planar 2D curves reconstructed from ZR-shapes may be self-intersecting.
% 	\end{enumerate}

	 Additionally considering closed curves invariant under change of initial point $z(s) \to z(s+s_0)$ (e.g. amorphous curves with no preassigned initial point) by defining this action of the unit circle $S^1\ni s_0$ on $\Sigma_{ZR}$ the \emph{invariant ZR--shape space} 
	\begin{eqnarray}\label{I-ZR-shape-def}
	\Sigma^{I}_{ZR} &:=&\Big(\Sigma_{ZR}\setminus\{0\}\Big)/S^1\,
	\end{eqnarray}
	is obtained.
% 	Representing ZR-shapes with phase angles the action can be expressed as 
% 	\begin{eqnarray*}\begin{array}{rcl}
%  	t &\stackrel{t_0}{\mapsto}& s:=t+t_0\\
% 	 \theta(t) &\mapsto& \widetilde{\theta}(t) := \theta(t +t_0)-\theta(t_0)\\
% 	A+\sum_{n=1}^{\infty} A_n \cos(nt-\alpha_n) &\mapsto& \widetilde{A} + \sum_{n=1}^{\infty} A_n \cos\big(nt - (\underbrace{\alpha_n -nt_0}_{=\widetilde{\alpha}_n})\big)
% 	\end{array}
% 	\end{eqnarray*}
% 	to obtain the \emph{invariant amplitudes} $A_1,A_2,\ldots$ and the \emph{phase angles} $\alpha_1,\alpha_2,\ldots$ reflecting additionally orientation and initial point. Of course, $A$ is determined by amplitudes and phase angles. \cite{ZR72} show that the Fourier basis is in a way optimal to obtain such a simple representation of the group action. 
	%Setting $\arccos\alpha_n = a_n$ note that the 

	Since rotation and parameter shift are equivalent for circles, the corresponding invariant ZR--shape is thus a singularity of $\Sigma^{I}_{ZR}$, in fact its only singularity. 

\paragraph{Parallel Transport on $\Sigma_{ZR}$ and $\Sigma^{I}_{ZR}$}

	Geodesics on $\Sigma_{ZR}$ as well as on $\Sigma^{I}_{ZR}$ between two given points can be computed via a technique called \emph{geodesic shooting}, cf. \cite{MTY06} as well as \cite{KSMJ04}, or much faster via a	 variational approach \cite{SCC06}. Since $\Sigma_{ZR}\hookrightarrow \lw^2$ is a Riemannian immersion with the global implicit definition $\Psi=0$ we have that the normal space in $S_{ZR}$ at $\theta\in\Sigma_{ZR}$ is spanned by $V_1(\theta)=s\mapsto\cos\left(\theta(s) +s\right)$ and $V_2(\theta)=s\mapsto \sin\big(\theta(s)+s\big)$. Orthogonalization yields the base
% 	$$ W_1(\theta) = \frac{V_1(\theta)}{\|V_1(\theta)\|},~~W_2(\theta) = \frac{V_2(\theta) - \langle V_2(\theta),V(\theta)\rangle\,V(\theta)}{\|V_2(\theta) - \langle V_2(\theta),V(\theta)\rangle\,V(\theta)\|}\,.$$
	$$ W_1 := \frac{V_1}{\|V_1\|},~~W_2 := \frac{V_2 - \langle V_2,W_1\rangle\,W_1}{\|V_2 - \langle V_2,W_1\rangle\,W_1\|}\,.$$
	As a consequence of Theorem  \ref{imm-parallel-th} we have
	\begin{Th}\label{ZR-Sigma-par-thm} A vector-field $W(t)$ in $\Sigma_{ZR}$ is parallel along a geodesic $\gamma$ in $\Sigma_{ZR}$ if and only if it satisfies the linear differential equation
	{\footnotesize \begin{eqnarray}\label{parallel-eq-ZR}
 	 \lefteqn{\dot{W}(t)~=~}\\\nonumber &-\Big(\frac{d}{dt}\,W_1\big({\gamma}(t)\big)\otimes W_1\big({\gamma}(t)\big) + \frac{d}{dt}\,W_2\big({\gamma}(t)\big)\otimes W_2\big({\gamma}(t))\Big) ~{W}(t)\,.
 	\end{eqnarray}}
% 	\begin{eqnarray}\label{parallel-eq-ZR}
%  	\nonumber \lefteqn{\dot{W}(t)~=~}\\\nonumber &-\Big(\frac{d}{dt}\,W_1\big({\gamma}(t)\big)\otimes W_1\big({\gamma}(t)\big)\\
% 	&\hspace{0.5cm} + \frac{d}{dt}\,W_2\big({\gamma}(t)\big)\otimes W_2\big({\gamma}(t))\Big) ~{W}(t)\,.
% 	\left(\id_{\lw^2} -  W_1\big({\gamma}(t)\big)\otimes W_1\big({\gamma}(t)\big) -  W_2\big({\gamma}(t)\big)\otimes W_2\big({\gamma}(t))\right) ~\dot{W}(t)&=&0\,.
%	 \nonumber \lefteqn{\dot{W}(t)~=~}\\ & -\left\langle {W}(t), \frac{d}{dt}\,W_1\big(\gamma(t)\big)\right\rangle\, W_1\big(\gamma(t)\big) - \left\langle {W}(t),\frac{d}{dt}\, W_2\big(\gamma(t)\big)\right\rangle\, W_2\big(\gamma(t)\big)\,.
% 	 \dot{W}(t) = \left\langle \dot{W}(t), W_1\big(\gamma(t)\big)\right\rangle\, W_1\big(\gamma(t)\big) + \left\langle \dot{W}(t), W_2\big(\gamma(t)\big)\right\rangle\, W_2\big(\gamma(t)\big)\,.
% 	\end{eqnarray}
	\end{Th}
	In practice, (\ref{parallel-eq-ZR}) can be solved numerically by orthogonally projecting to $T_{\gamma(t)}\Sigma_{ZR}$ in every iteration step. 
% 	In coordinates (\ref{parallel-eq-ZR}) is a linear (infinite) system of first order ODEs. Since in practice, only finitely many Fourier coefficients of $\gamma(t)$ are used, (\ref{parallel-eq-ZR}) can be numerically integrated with standard numerical methods (\cite{?}).

%	\paragraph{Parallel Transport on $\Sigma^{(1)}_{ZR}$.} 
	We now turn to the submersion (\ref{I-ZR-shape-def}). The vertical space at 
	$$\theta(s) = x_0+\sum_{n=1}^{\infty} \big(x_n \cos(ns) + y_n \sin(ns)\big) \in \Sigma_{ZR}$$
	is spanned (if convergent) by the single vertical unit length direction 
	$$\frac{\theta'(s)}{\|\theta'(s)\|} = \frac{\sum_{n=1}^{\infty} n\big(-x_n \partial_{y_n} + y_n \partial_{x_n}\big)}{\sqrt{\sum_{n=1}^{\infty} n^2(x_n^2+y_n^2)}}~\sqrt{2}\,.$$ 
	The exterior derivative of its dual is hence
	\begin{eqnarray}\label{ext-der-omeg-ZR}
	 	\left.\begin{array}{rcl}
	\frac{d\omega}{\sqrt{2}}&=& -2\,\frac{\sum_{n=1}^{\infty}ndx^n \wedge dy^n}{\sqrt{\sum_{n=1}^{\infty} n^2(x_n^2+y_n^2)}}\,\\&&~~ -\,\frac{\sum_{n\neq n'}y_nx_{n'}n n'(n'dx^{n'} \wedge dx^{n}-ndy^n \wedge dy^{n'})}{\sqrt{\sum_{n=1}^{\infty} n^2(x_n^2+y_n^2)}^{~3}}\\
	&&-\,\frac{\sum_{n,n'}nn'(ny_ny_{n'}+n'x_nx_{n'}) dy^{n} \wedge dx^{n'}}{\sqrt{\sum_{n=1}^{\infty} n^2(x_n^2+y_n^2)}^{~3}}\end{array}\right\}\,.
	\end{eqnarray}
	In conjunction with Theorem \ref{imm-parallel-th}, Theorem \ref{parallel_riem_subm:thm} and Theorem \ref{ZR-Sigma-par-thm} one obtains after a tedious computation

	\begin{Th}\label{ZR1-Sigma-par-thm}
	The vector-field 
	$$ W(t) = u_0(t)\partial_{x_0}+\sum_{n=1}^\infty \big(u_n(t)\partial_{x_n} + v_n(t)\partial_{y_n}\big)$$
	is a horizontal lift to the top space $\Sigma_{ZR}$ of the bottom space  parallel transport along a geodesic in $\Sigma^{I}_{ZR}$
	$$\gamma_s(t) = x_0(t)+\sum_{n=1}^\infty \big(x_n(t)\cos(ns) + y_n(t)\sin(ns)\big)$$
	horizontal in $\Sigma_{ZR}$ if and only if it satisfies the linear differential equation
	{\footnotesize\begin{eqnarray*}%\label{parallel-eq-ZR1}\nonumber
 	 \lefteqn{\dot{W}(t)~=~}\\ %\nonumber
	&&-\Big(\frac{d}{dt}\,W_1\big({\gamma}_s(t)\big)\otimes W_1\big({\gamma}_s(t)\big) + \frac{d}{dt}\,W_2\big({\gamma}_s(t)\big)\otimes W_2\big({\gamma}_s(t))\Big) ~{W}(t)
	\\ %\nonumber
	&&+\,\frac{\gamma'_s(t)}{2\|\gamma'_s(t)\|^4}\Big(\big\langle \gamma'_s(t)\otimes\gamma''_s(t), W(t)\otimes \dot{\gamma}_s(t)\big\rangle\\%\nonumber
	 &&\hspace{3cm} - \big\langle \gamma'_s(t)\otimes\dot{\gamma}_s(t), W(t)\otimes {\gamma}''_s(t)\big\rangle\Big)\\
	&&-\,\frac{\gamma'_s(t)}{\|\gamma'_s(t)\|^2} \big\langle \dot{\gamma}'_s(t),W(t)\rangle
	\,.
	\end{eqnarray*}}with the derivatives defined as
	$$\begin{array}{rcl}
  	\gamma'_s(t) &=& \sum_{n=1}^{\infty} n\big(-x_n(t) \partial_{y_n} + y_n(t) \partial_{x_n}\big)\,,\\ \dot{\gamma}_s(t) &=& \dot{x}_0(t)\partial_{x_0} + \sum_{n=1}^{\infty} \big(\dot{x}_n(t) \partial_{x_n} + \dot{y}_n(t) \partial_{y_n}\big)\,,\\
	\dot{\gamma}'_s(t)&=& \sum_{n=1}^{\infty} n\big(-\dot{x}_n(t) \partial_{y_n} + \dot{y}_n(t)
	 \partial_{x_n}\big)\,,\\
  	\gamma''_s(t) &=& -\sum_{n=1}^{\infty} n^2\big(x_n(t) \partial_{x_n} + y_n(t) \partial_{y_n}\big)\,,	
	\end{array}$$
	if convergent, and the inner product defined by %\mbox{ and}\\
	$$\langle E_i\otimes E_j,E_k\otimes E_l\rangle ~:=~ \delta_{(i,j),(k,l)}$$
	for an orthogonal system $E_j$ and index set $J\ni j$.    
	\end{Th}

	In practice, convergence of the series for the derivates is not an issue since as remarked earlier, computations are carried out using only finitely many Fourier coefficients.
% 	\begin{Th}[Part of Old Version]\label{ZR1-Sigma-par-thm}
% 	\begin{eqnarray}\label{parallel-eq-ZR1}\nonumber
% 	 \dot{W}(t) &=& \left\langle \dot{W}(t), W_1\big(\gamma(t)\big)\right\rangle\, W_1\big(\gamma(t)\big) + \left\langle \dot{W}(t), W_2\big(\gamma(t)\big)\right\rangle\, W_2\big(\gamma(t)\big) \\&&\hspace{4.5cm}-\, d\omega\big(\dot{\gamma}(t),W(t)\big) \,V(t)\,,
% 	\end{eqnarray}
% 	with the vertical field $V(t)$ and the value of the exterior derivative of its dual $\omega$:
% 	$$\begin{array}{rcl}
% 	V(t)&=& \frac{\sum_{n=1}^{\infty} n\big(-x_n(t) \partial_{y_n} + y_n \partial_{x_n}\big)}{\sqrt{\sum_{n=1}^{\infty} n^2\big(x_n(t)^2+y_n(t)^2\big)}}\,,\\
% 		2d\omega\big(\dot{\gamma}(t),W(t)\big) &=&-2\,\frac{\sum_{n=1}^{\infty}n(\dot{x}_nv_n - u_n\dot{y}_n)}{\sqrt{\sum_{n=1}^{\infty} n^2(x_n^2+y_n^2)}}\,\\&&~~ -\,\frac{\sum_{n\neq n'}y_nx_{n'}n n'\big(n'(\dot{x}_{n'}u_n - u_{n'}\dot{x}_n)-n(\dot{y}_{n}v_{n'} -v_{n}\dot{y}_{n'})\big)}{\sqrt{\sum_{n=1}^{\infty} n^2(x_n^2+y_n^2)}^{~3}}\\
% 	&&-\,\frac{\sum_{n,n'}nn'(ny_ny_{n'}+n'x_nx_{n'}) (\dot{y}_nu_{n'}-v_n\dot{x}_{n'})}{\sqrt{\sum_{n=1}^{\infty} n^2(x_n^2+y_n^2)}^{~3}}\,.
% 	\end{array}$$
% 	\end{Th}

\section{PT for Kendall's Shape Spaces}\label{Kendall_scn}
 	Kendall's landmark based similarity shape analysis is based on \emph{configurations} consisting of $k\geq m+1$ labelled vertices
        in $\mathbb R^m$  called \emph{landmarks} that do not all
        coincide. A configuration 
	$$x = (x^1,\ldots,x^k) = (x^{ij})_{1\leq i\leq m,1\leq j\leq k}$$ 
	is thus an element of the space $M(m,k)$ of matrices with $k$ columns, each an $m$-dimensional landmark vector. Disregarding center and size, these configurations are mapped to the \emph{pre-shape sphere} 
      $$ M = S^{k}_m := \{p \in M(m,k-1) \colon \|p\| = 1\}\,,$$
	where  $\|p\|^2 =  \langle p, p \rangle $ and  $ \langle p, v \rangle := \tr(pv^T)$ is the standard Euclidean product. 
      	This can
		be done by, say,  multiplying by a sub-Helmert matrix,
		cf. \cite{DM98} for a detailed discussion of this and other
		normalization methods. The canonical Riemannian immersion  $S_m^k\hookrightarrow M(m,k-1)$ comes with a global implicit definition $\Psi(x) = \|x\|-1 =0$ giving rise to a single normal field. Hence, using (\ref{Eucl_nabl_eq}) and (\ref{imm-cov-der-eq}), the covariant derivatives along a curve $\gamma$ on $S_m^k$ relate as
		\begin{eqnarray}\label{spher-cov-der}
		 \nabla^{S_m^k}_{\dot{\gamma}(t)} W(t) &=& \dot{W}(t) - \big\langle \dot{W}(t),\gamma(t)\big\rangle \, \gamma(t)\,
		\end{eqnarray}
	As a consequence of (\ref{spher-cov-der}) and (\ref{parallel_eq}), unit-speed geodesics on $S_m^k$ are great circles of form $\gamma(t) = x\cos t + v\sin t$ with $x,v\in S_m^k$ and $\langle x,v\rangle=0$.

	In order to filter out rotation information define the regular part $(S^{k}_m)^* := \{x\in S_k^k: \rank(x) > m-2\}$ (an open dense subset of $S_m^k$) and a
      	smooth and free action of $SO(m)$ by the usual matrix multiplication
        $(S^{k}_m)^* \stackrel{g}{\to} (S^{k}_m)^* : p \mapsto gp$
        for $g \in SO(m)$. The
        orbit 
        $\pi(p) = \{gp \colon g \in SO(m)\}$ is the \emph{Kendall shape} of $p \in
        S^{k}_m$ and the quotient
	\begin{eqnarray}\label{def-shape-space}
        \pi : (S^{k}_m)^* ~\to~ (\Sigma^{k}_m)^* &:=& (S^{k}_m)^*/SO(m)
	\end{eqnarray}
	is called \emph{Kendall's shape space}. Note that projecting from the entire pre-shape sphere $S_m^k$ would have led to a non-manifold quotient ($m\geq 3$), which is usually called Kendall's shape space. We consider here only the regular part such that (\ref{def-shape-space}) is a Riemannian submersion. With the orthogonal decomposition $\frak{gl}(m) = \frak{o}(m) \oplus SM(m)$ of the Lie algebra $\frak{gl}(m) = M(m,m)$, the Lie algebra $\frak{o}(m)$ of skew-symmetric matrices in $\frak{gl}(m)$ and the vector-space of symmetric matrices $SM(m)$ in $\frak{gl}(m)$ we have the following orthogonal tangent space decomposition for $x\in S_m^k$, 
	cf. \citet[p.109]{KBCL99}. 
	\begin{eqnarray}\label{soms}
		&\begin{array}{rcccc}
		\frak{gl}(m) &=& {\frak o}(m) &\oplus& SM(m)\\
		&                     & \downarrow \cdot x &&\uparrow \cdot x^T\\
		T_xS^k_m\oplus N_x S_m^k&=&T_x\pi(x)& 
		\oplus &\overbrace{ H_xS^k_m\oplus N_x S_m^k}
		\end{array}
	\end{eqnarray}
	For $x\in (S_m^k)^*$ both mappings are surjective and $H_xS_m^k = T_{\pi(x)}(\Sigma_m^k)^*$.
	In order to compute the horizontal lift of bottom space parallel transport as in Theorem \ref{ZR1-Sigma-par-thm}, we need an orthonormal base for the $\big(m(m-1)/2\big)$-dimensional vertical space $T_x\pi(x)$ and the exterior derivative of its duals. From (\ref{soms}) we have at once a (in general  not-orthogonal) base $\{e_{ij}x : 1\leq i<j\leq m\}$ with base system	
	{\footnotesize $$e_{ij} = \left(\varepsilon^{\alpha,\beta}\right)_{1\leq \alpha, \beta\leq m}\mbox{ with } \varepsilon^{\alpha,\beta} =\left\{\begin{array}{lcl}1&\mbox{ for } \alpha = i,\beta=j\\
	-1 &\mbox{ for } \alpha = j,\beta=i\\
	0&\mbox{ else }\end{array}\right.$$}of $\frak{o}(m)$. From the former obtain an o.g. base system $\{V_{ij}(x): 1\leq i<j\leq m\}$ of $T_x\pi(x)$ through Gram-Schmidt orthogonalization and let  $\omega_{ij}(x)$ be the one-form dual to $V_{ij}(x)$ ($1\leq i<j\leq m$). For $m=2$ there is a single vertical unit-direction $V_{12}(x)$. For $m>2$, however, $d\omega_{12}(x)$ is about as complicated as (\ref{ext-der-omeg-ZR}), and even for $m=3$, the other two derivatives and their application to vector-fields result in expressions too lengthy to be written down, cf. also the rather complicated examples %for parallel transport given 
	in \cite{L03}. Using a computer algebra program, however, these expressions and their respective values can be easily computed by symbolic differentiation. Hence, %(\ref{O_Neil_parallel_formula}) and 
	Theorem \ref{parallel_riem_subm:thm} %(\ref{single_vert_parallel_transo}) 
	and (\ref{spher-cov-der}) 
	yield the following Theorem. The special case $m=2$ is taken from \citet[Theorem A.6]{HHM09}, cf. also \citet[Theorem 2]{L03}.
	\begin{Th}\label{Parallel-transp-Kendall}
	A vector-field 
	$ W(t)$ 
	is a horizontal lift to the top space of the bottom  space $(\Sigma^k_m)^*$ parallel transport along a geodesic 
	$\gamma(t) = x\cos t +v\sin t$
	horizontal in $(S_m^k)^*$ if and only if it satisfies the ODE
	\begin{eqnarray*}
 	\lefteqn{\dot{W}(t) ~=~ 
 	\left\langle \dot{W}(t), \gamma(t)\right\rangle\, \gamma(t)}
	 \\ &&\hspace{0cm}-\sum_{1\leq i<j\leq m}
	d\omega_{ij}\big(\gamma(t)\big)\Big(\dot{\gamma}(t),W(t)\Big)\, V_{ij}\big(\gamma(t)\big)\,.	
	\end{eqnarray*}
	%\end{eqnarray*}
	For $m=2$ this ODE has the explicit solution
	\begin{eqnarray*} W(t) ~=~ W(0) - \Big(\big\langle W(0),v\big\rangle\, v + \big\langle W(0),e_{12}v\big\rangle\, e_{12}v\Big) \\+ \big\langle W(0),v\big\rangle\, \dot{\gamma}_{z,v}(t) + \big\langle W(0),e_{12}v\big\rangle\, e_{1,2}\dot{\gamma}_{z,v}(t)\,.\end{eqnarray*}
	\end{Th}

	\begin{figure}[h!]
 	\begin{minipage}[t]{0.45\textwidth}
	 \subfigure[{\it From rectangle $\sigma_1$ to the regular hexagon $\sigma_3$.}]{\label{Def_R1_H_ZR_fig} 
	\includegraphics[angle=-90,width=0.8\textwidth]{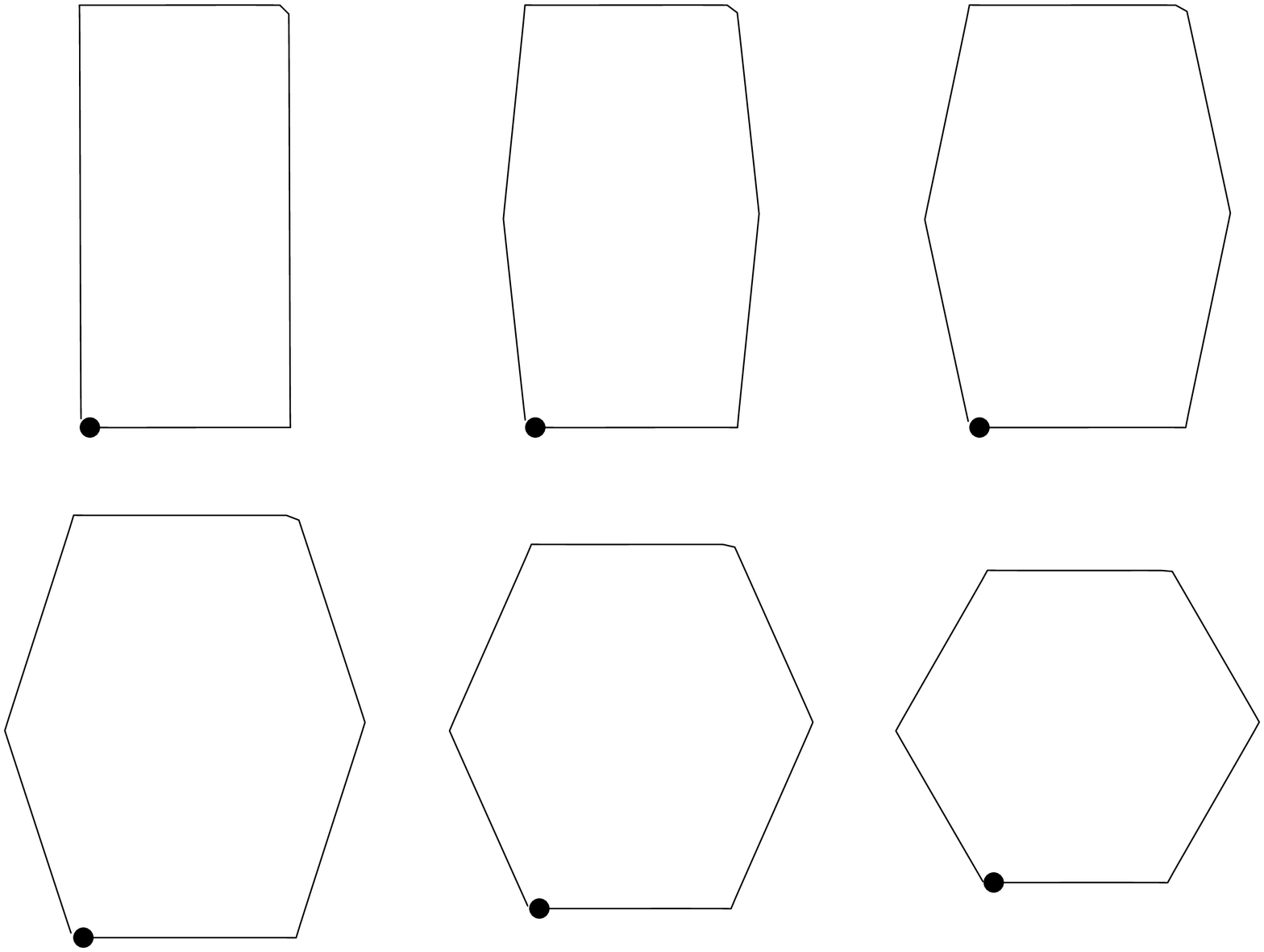}}
% 	\begin{minipage}[t]{0.05\textwidth}\hfill \end{minipage}
	 \subfigure[{\it From rectangle $\sigma_1$ to rectangle $\sigma_2$.}]{\label{Def_R1_R2_ZR_fig} 
	\includegraphics[angle=-90,width=0.8\textwidth]{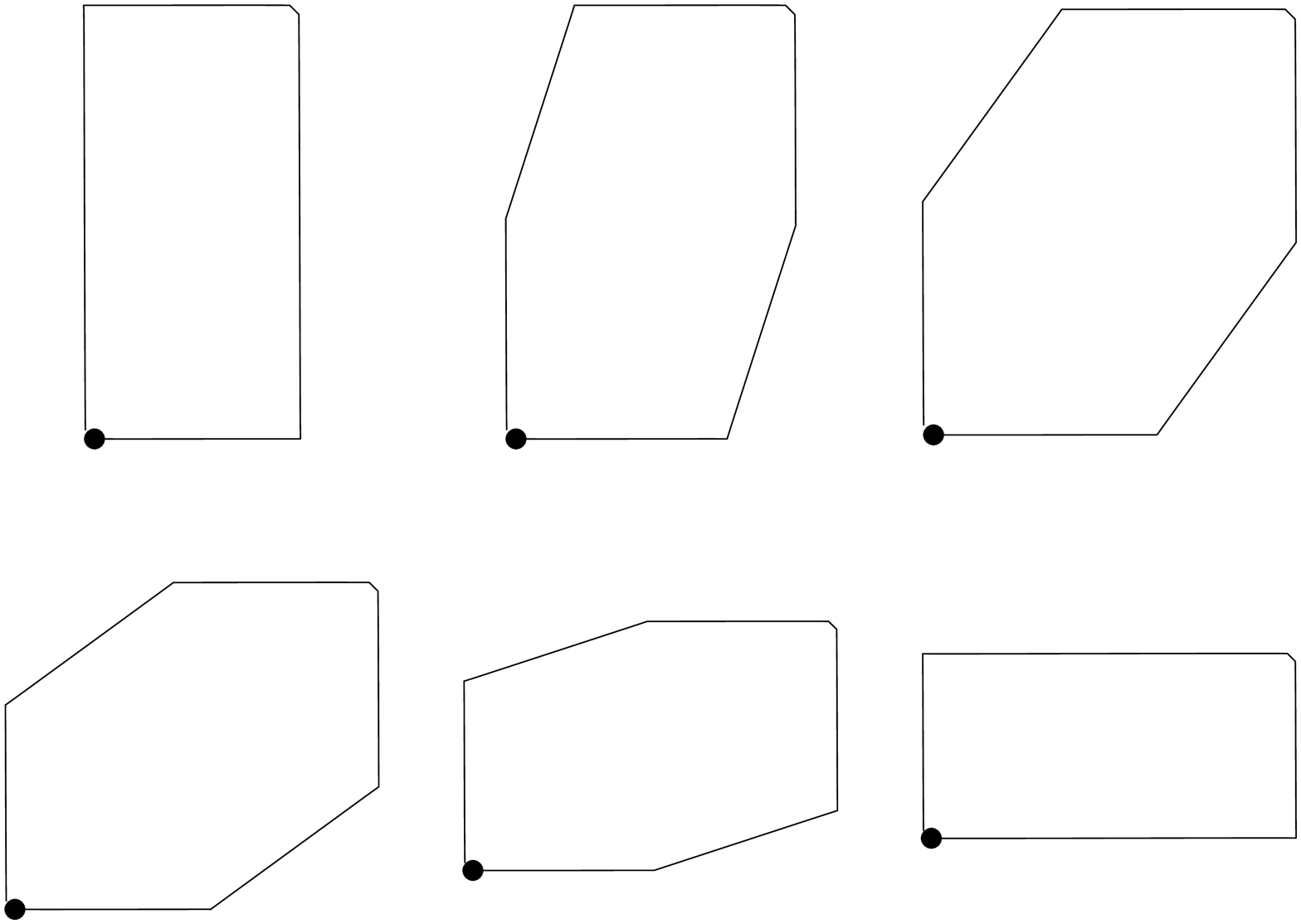}}
	 \subfigure[{\it From $\sigma_2$ to the parallel transplant of $\sigma_3$ from Figure \ref{Def_R1_H_ZR_fig}.}]{\label{Def_R2H_para_ZR_fig} 
	\includegraphics[angle=-90,width=0.8\textwidth]{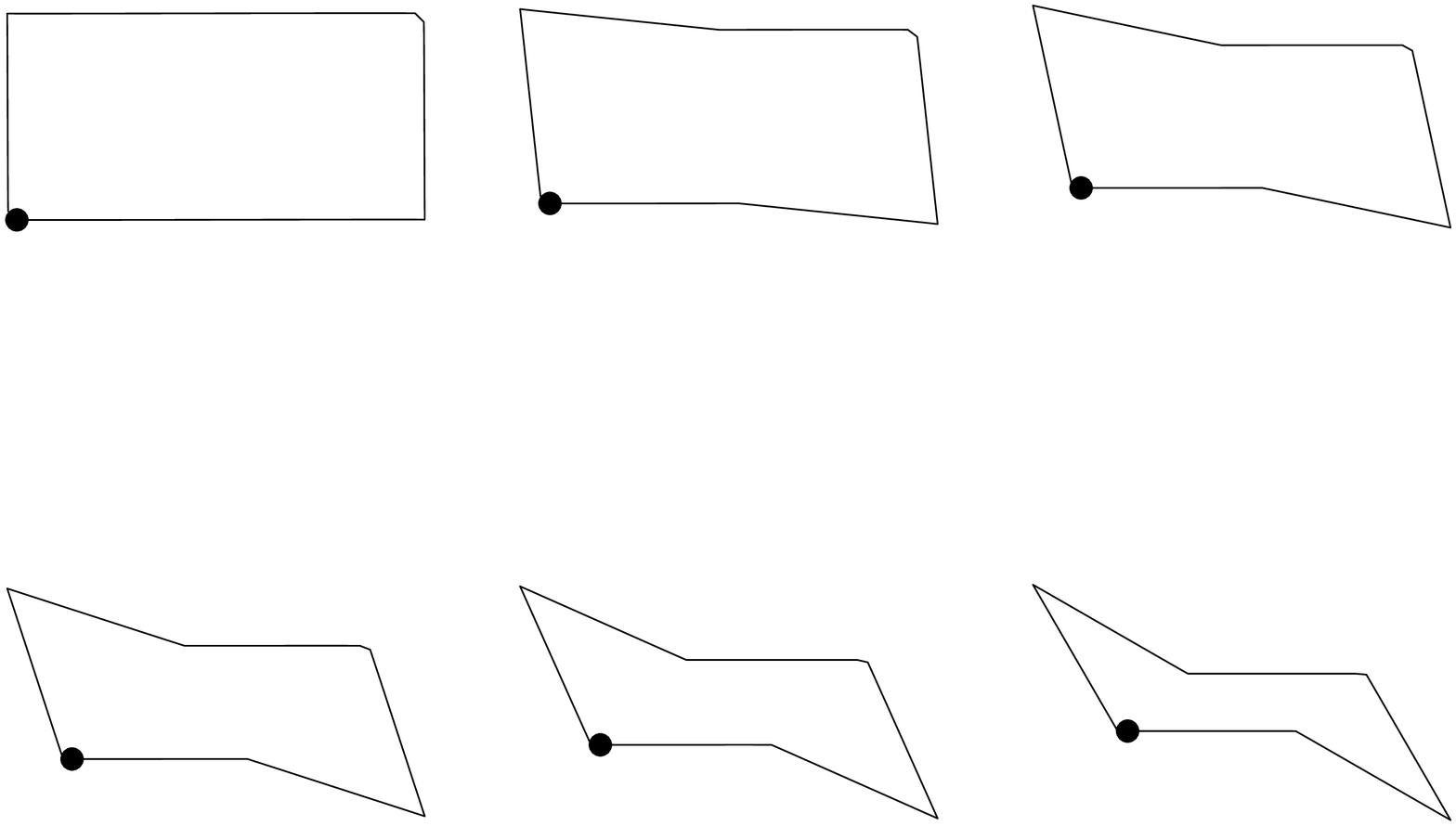}}
% 	\begin{minipage}[t]{0.05\textwidth}\hfill \end{minipage}
% 	\begin{minipage}[t]{0.45\textwidth}
	\caption{\it Equidistant deformation along geodesics in $\Sigma_{ZR^2}$. The bullet marks the pre-assigned initial point. Figure \ref{Def_R2H_para_ZR_fig} depicts the parallel transplant to $\sigma_2$ of the geodesic from Figure \ref{Def_R1_H_ZR_fig} along the geodesic depicted in Figure \ref{Def_R1_R2_ZR_fig}.\label{ZR_Hex_fig}}
         \end{minipage}
	\begin{minipage}[t]{0.05\textwidth}\hfill\end{minipage}
	\begin{minipage}[t]{0.45\textwidth}
	 \subfigure[{\it From rectangle $\sigma_1$ to the regular hexagon $\sigma_3$.}]{\label{Def_R1_H_K_fig} 
	\includegraphics[angle=-90,width=0.8\textwidth]{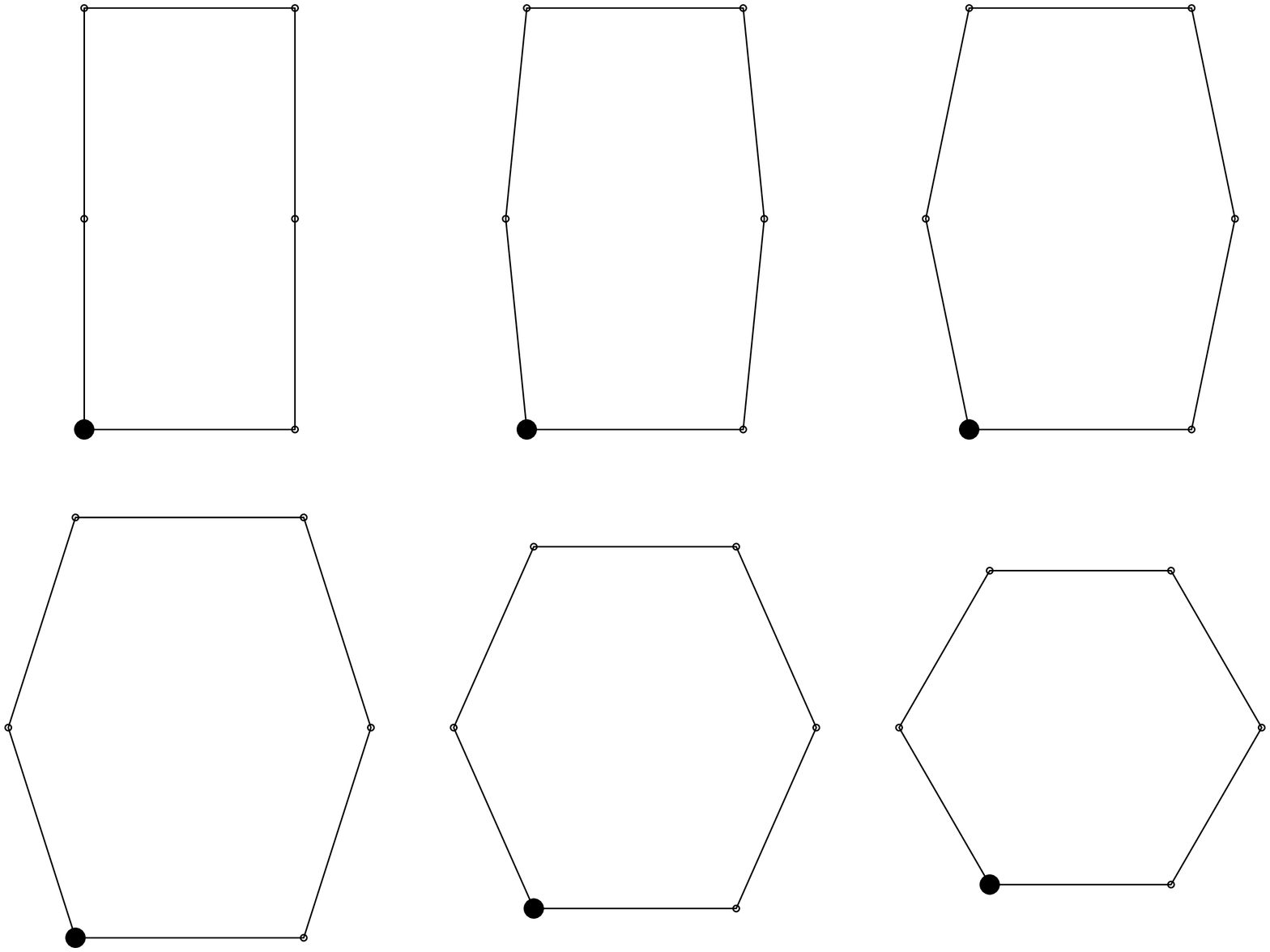}}
% 	\begin{minipage}[t]{0.05\textwidth}\hfill \end{minipage}
	 \subfigure[{\it From rectangle $\sigma_1$ to rectangle $\sigma_2$.}]{\label{Def_R1_R2_K_fig}
	 \includegraphics[angle=-90,width=0.8\textwidth]{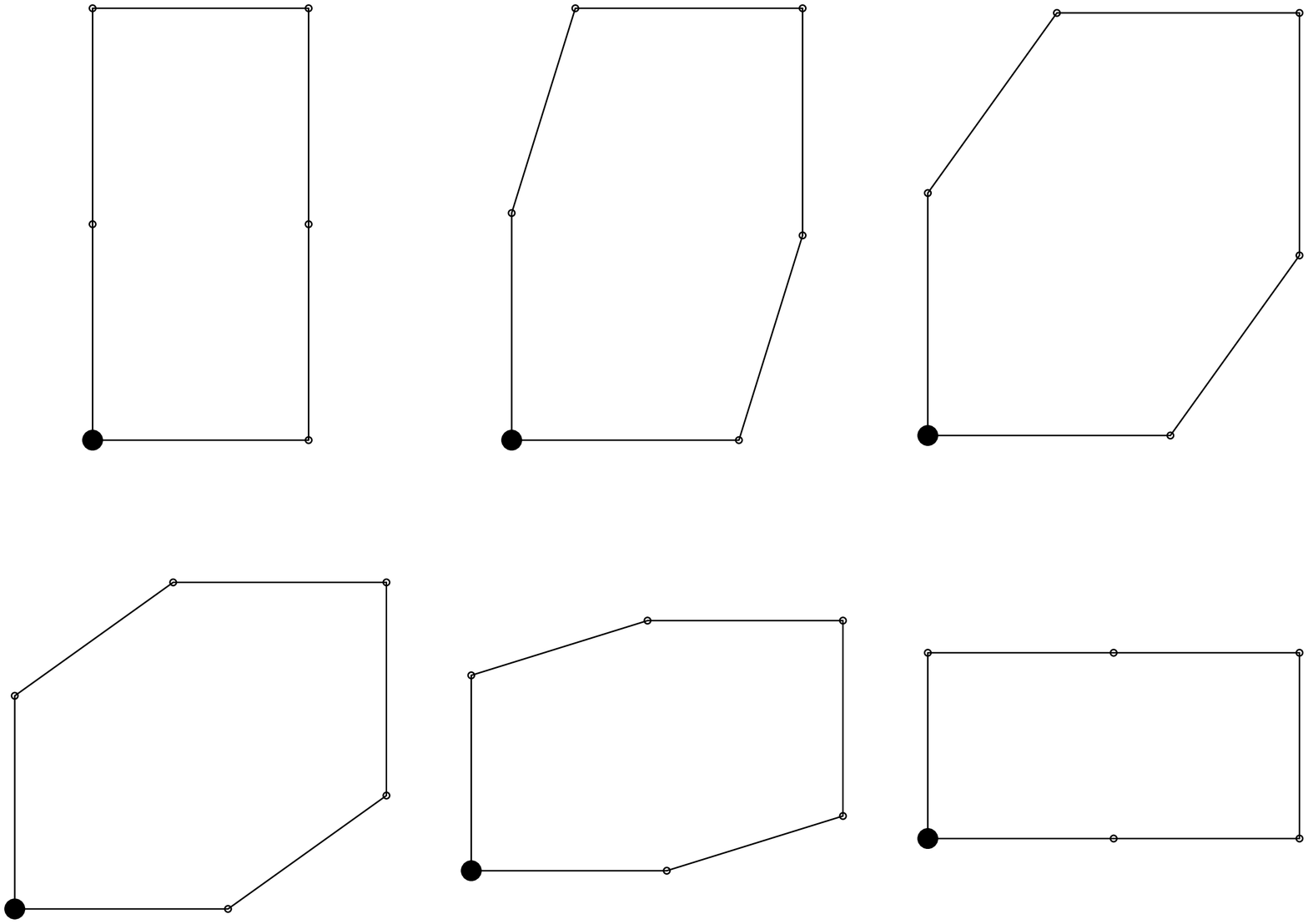}}
	 \subfigure[{\it From $\sigma_2$ to the parallel transplant of $\sigma_3$ from Figure \ref{Def_R1_H_K_fig}.}]{\label{Def_R2H_para_K_fig} 
	\includegraphics[angle=-90,width=0.8\textwidth]{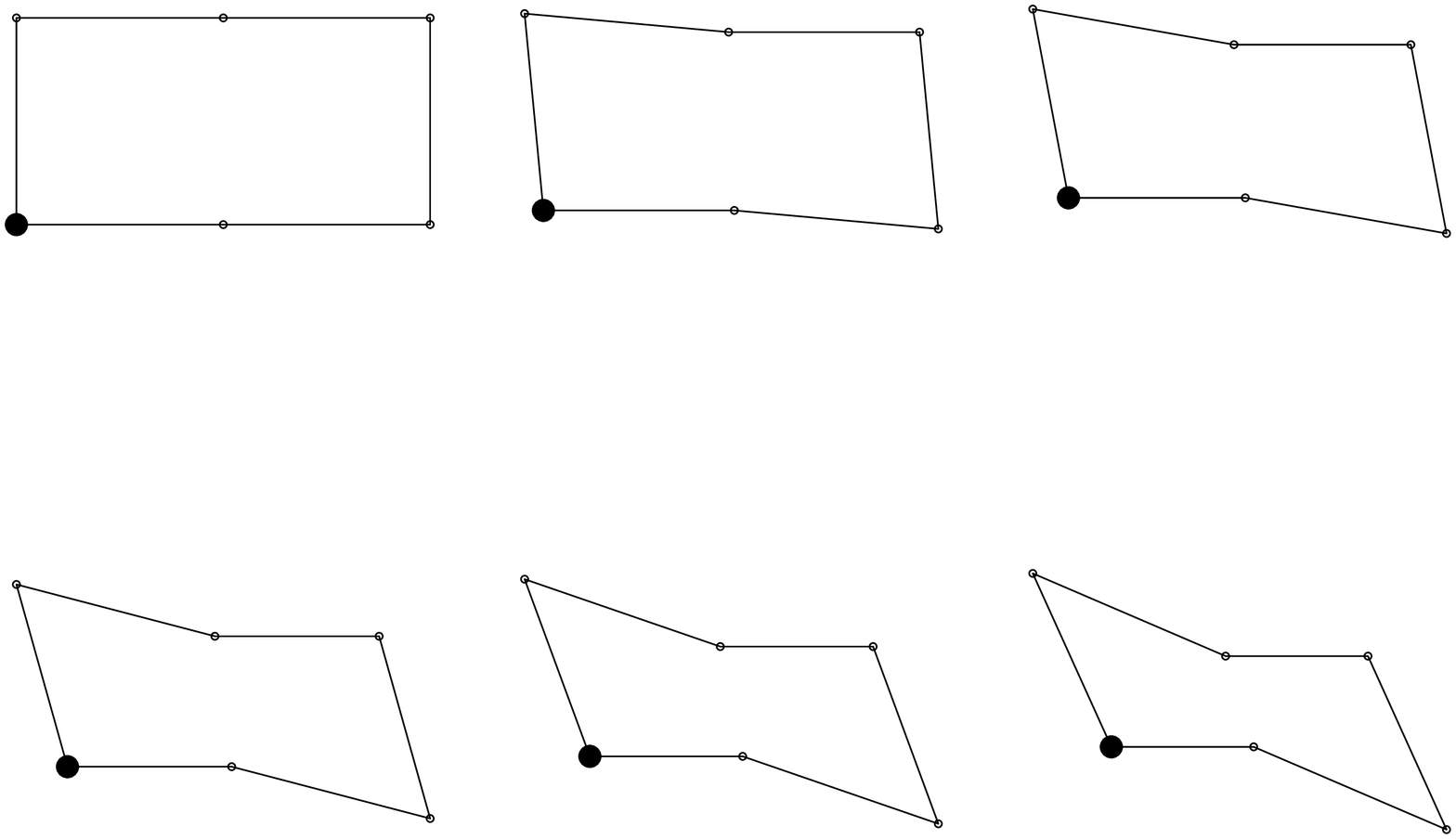}}
% 	\begin{minipage}[t]{0.05\textwidth}\hfill \end{minipage}
% 	\begin{minipage}[t]{0.45\textwidth}
	\caption{\it Equidistant deformation along geodesics in Kendall's landmark based shape space $\Sigma^6_2$. The kinks signify landmarks. Notation as in Figure \ref{ZR_Hex_fig}.\label{K_Hex_fig}}
         \end{minipage}
	\end{figure}

\section{Curves with Rotational Symmetry}\label{k-rot-inv:scn}

	Denote by $\Sigma_{ZR^k}$ the sub-space of closed curves with $k$-fold rotational symmetry. The following is an observation of \cite{ZR72}:

	\begin{Th} 
	 $\theta$ represents a closed curve with $k$-fold rotational symmetry $k>1$ if and only if $x_n,y_n =0$ for all $n\not\equiv 0$ mod $k$.
	\end{Th}

% 	As \cite{ZR72} observed,that $\theta$ represents a closed curve with $k$-fold rotational symmetry $k>1$ if and only if $x_n,y_n =0$ for all $n\not\equiv 0$ mod $k$. Denote by $\Sigma_{ZR^k}$ the sub-space of closed curves with $k$-fold rotational symmetry. One verifies at once that

	Since the arithmetic sum of closed curves with $k$-fold rotational symmetry is again of $k$-fold rotational symmetry we have at once:

	\begin{Cor} For each $k=2,3\ldots$, $\Sigma_{ZR^k}$ is a flat linear submanifold of $\Sigma_{ZR}$ and  parallel transport on $\Sigma_{ZR^k}$ is affine.
	\end{Cor}

% 
% \section{Comparing the `Same' Shape Deformation over Different Shape Spaces}

	Within the subspace  $\Sigma_{ZR^2}$ of closed curves with two-fold symmetry consider a simple example of three shapes: two rectangles $\sigma_1$ and $\sigma_2$ differing only by their initial point and a hexagon $\sigma_3$. The `same' (i.e. parallel) deformation from $\sigma_1$ to $\sigma_3$ is applied to $\sigma_2$.

	Figure \ref{ZR_Hex_fig} illustrates that deformation in the geometry of Zahn-Roskies' shape space, Figure \ref{K_Hex_fig} gives it in the geometry of Kendall's shape space. Comparing the respective Subfigures (a) and (b) over the different geometries shows that the geodesic deformation with fixed initial and terminal shape gives almost identical intermediate shapes. 
	
	In the respective Subfigures (c), the difference between `same' shape deformation over the two geometries is hardly notable in the beginning of the deformation. Near the end, however, it becomes notable: number the kinks (the landmarks in $\Sigma_2^6$) counterclockwise from 1 (bullet) to 6 and denote by $(i,j)$ the line connecting the $i$-th kink with the $j$-th kink. Then $(2,5)$ remains parallel to $(6,1)$ and $(3,4)$ in the Kendall geometry whereas in the Zahn-Roskies geometry it turns in direction beyond $(2,3)$ and $(5,6)$.

% 	The same is true for slight parallel deformation as seen in the first row of the respective Subfigures (c). For larger deformation as seen in the respective second rows differences between `same' shape deformation over the different geometries become notable. One might prefer Kendall's hexagonal shape space over Zahn-Roskies' shape space for reasons of symmetry. Although in both geometries opposite sides remain are parallel, the Kendall geometry features additional symmetry: lines connecting arbitrary bends (landmarks) remain parallel, in particular the line connecting the two middle bends is parallel to the top and bottom line. 

\section{Parallel Leaf Growth}\label{App:scn}

	Let us quickly overview some very recent developments within two millenia of research on plant form. With the application below in mind, we are interested in a flexible and realistic representation of leaf contour shape. Flexibility in this context means that we are looking for a model in which nature not only chooses values of parameters in a pre-defined parameter space but rather the parameter space itself. We caution that such a \emph{non-parametric} model may come at the cost that nature's parameters may not be simply geometrically interpretable 

	\emph{Parametric} in this sense are the well established models involving allometry, still of interest today: e.g.  \cite{Gur92}; \cite{Burton04}, or the superformula of \cite{Gielis03}. Also, models based on landmarks such as \cite{DPS87}; \cite{Jensen90} or \cite{JCM2002}, can be viewed as projecting nature to a pre-specified parameter space by leaving out the parts of the contour between landmarks. Note that parametric models are highly successful e.g. for plant classification, genetic hybrid identification, cf. \cite{JM05}, or in the \emph{Climate Leaf Analysis Multivariate Program (CLAMP)} of \citet{Wolfe93} which is fundamental to paleoclimate and present day climate reconstruction, cf. \cite{EBG00}. %And, in some cases it might not even be clear, where to place landmarks, cf. 

	On the other hand, models building on the shape spaces of Zahn and Roskies (cf. Section \ref{ZR-scn}) are \emph{non-parametric}, even though they are not entirely free of constraints: in view of landmark-based shape analysis (cf. Section \ref{Kendall_scn}), restricting to unit speed velocities translates into infinitesimally placed landmarks. The different geometry,  however, liberates from the necessity to identify homologous landmarks, by imposing infinitesimal uniform growth. The latter is certainly debatable.  Curiously, such non-parametric models introduced as \emph{eigenshape analysis} (building on \cite{Loh83}) have initially stirred  controversy because parameters were not simply geometrically interpretable and because with lacking initial point, registration was not satisfactory, cf. \cite{R86}. While the former is precisesly a desired feature, %and not a bug of this model, 
	introducing the geometric concept of the quotient $\Sigma^I_{ZR} = \big(\Sigma_{ZR}\setminus\{0\}\big)/S^1$ by \cite{KSMJ04} settles the latter objection. It seems, however, that the natural non-Euclidean geometry of $\Sigma_{ZR}$ is not fully realized in the community, cf. \cite{Ray92}; \cite{KGS07}; \cite{Hearn09}.

	For sake of completeness, even though not practicable for our purpose because of high sensitivity to boundary noise, let us briefly mention a third approach of shape modeling based on the leaf's vein structure. With methods for automated venation extraction available (cf. \cite{FuChi06}), although computationally much more challenging than contour extraction, \cite{LZG09} link vein structure to the concept of shape spaces by \cite{BN78} based on medial skeletons, cf. also \cite{PSSDZ03}. Undoubtedly, modeling the vein structure gives deep insight into physiological, hydraulical and biomechanical aspects of leaf formation. Parameter spaces thus obtained should be closest to nature in the above sense. Current research in venation patterns, however, shows that leaf shape diversification is still poorly understood (e.g. \cite{NPT07}). 

	Obviously, for our purpose of modeling entire leaf contours while being as non-parametric as possible, the space $\Sigma_{ZR}$ suits ideally. For the problem at hand there is no need for pre-registration as the leaves in question are naturally aligned by petiole (the base point where the stalk enters the blade forming the main leaf vein) and apex (the terminal point of the main vein, usually the leaf tip) location. One could almost equivalently align by petiole location and the initial direction of the main leaf vein.

	\begin{figure}[h!]
	 \includegraphics[angle=-90,width=0.45\textwidth]{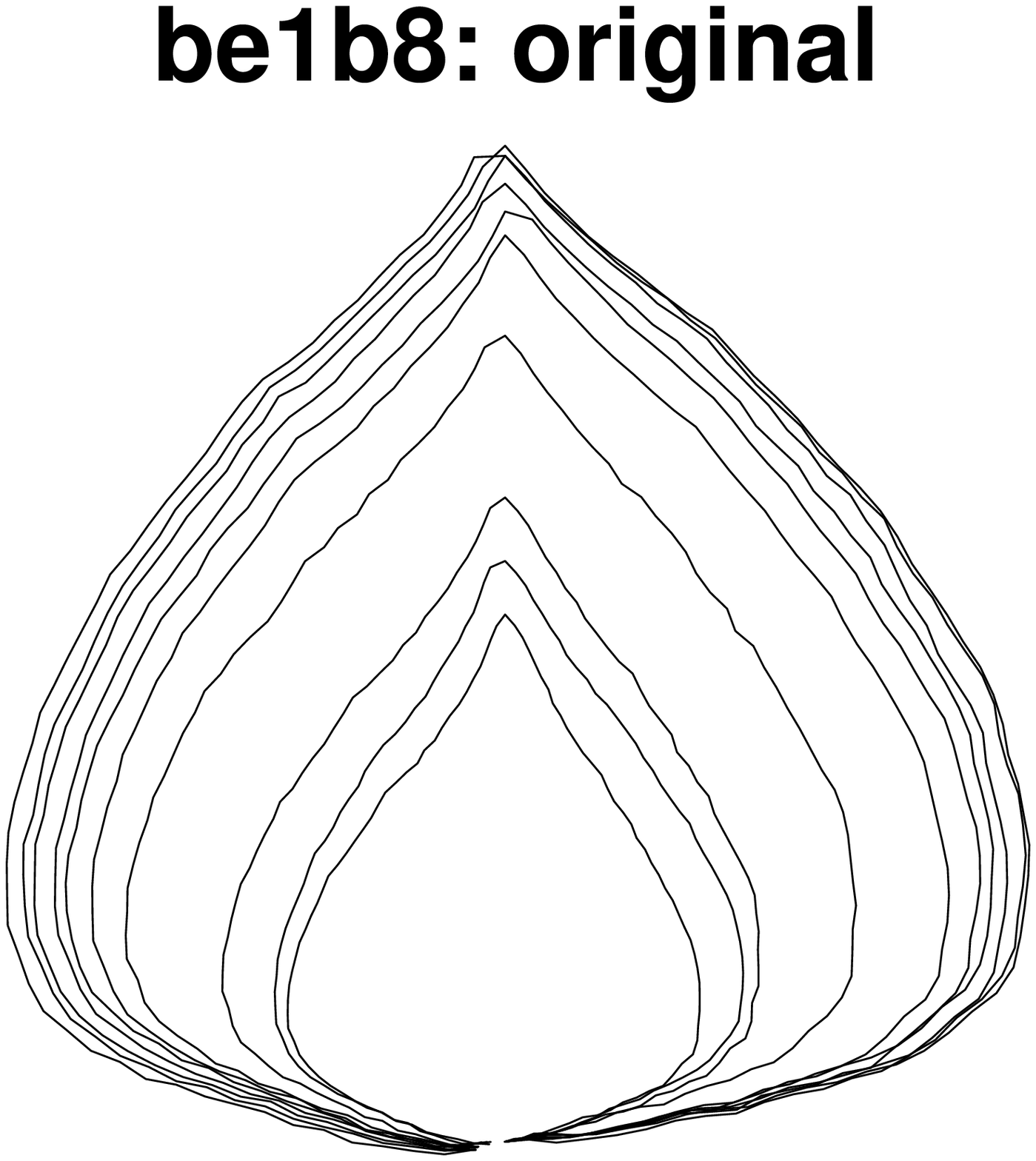}
    	\includegraphics[angle=-90,width=0.45\textwidth]{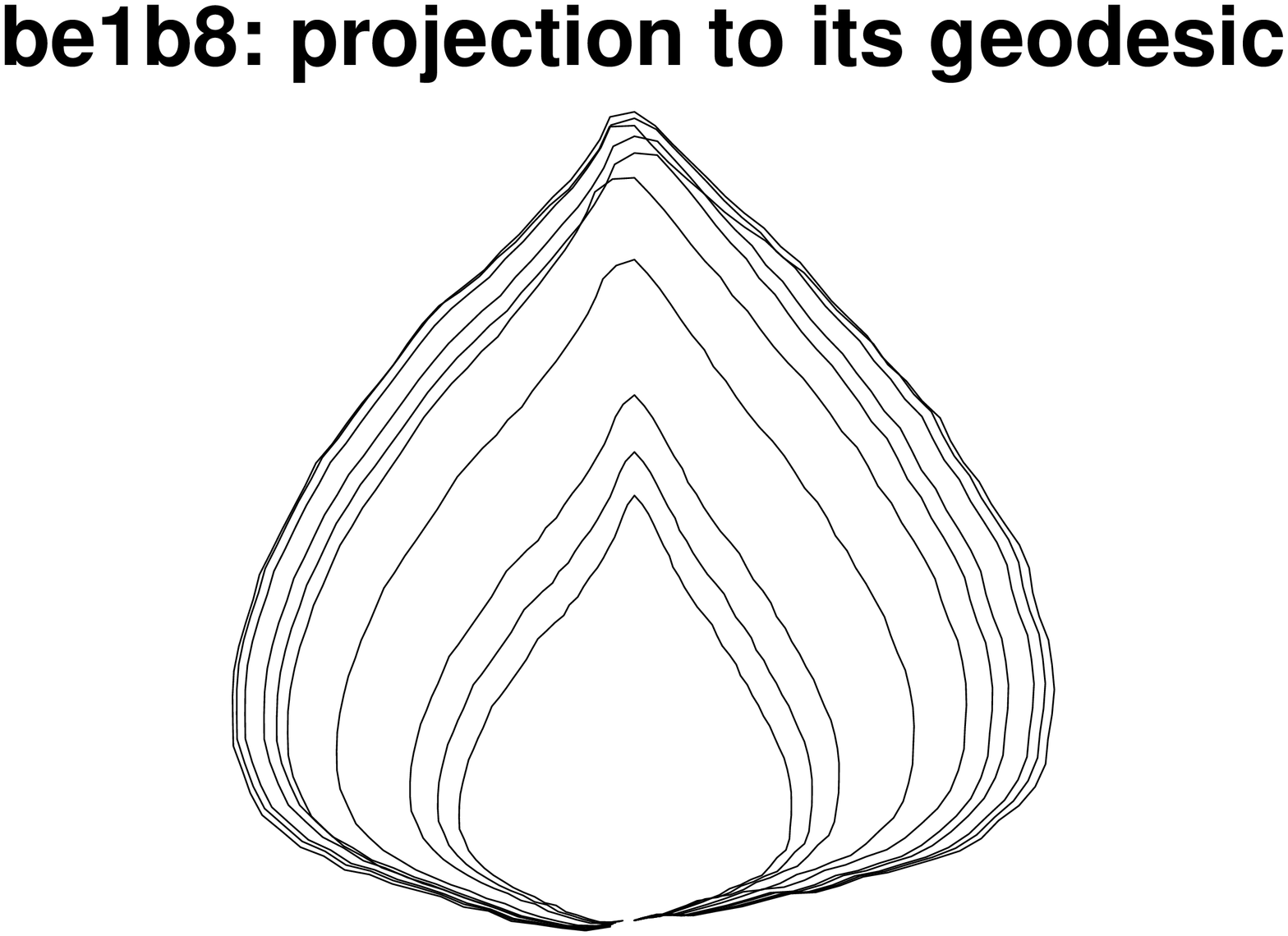}
	\caption{\it Shape evolution of a black poplar leaf over two weeks. Left: original contours. Right: contours obtained from projecting to geodesic evolution. \label{ori-geod-be1b8.eps}}
	\end{figure}

	\begin{figure}[h!]
	\centering
	 \includegraphics[angle=-90,width=0.45\textwidth]{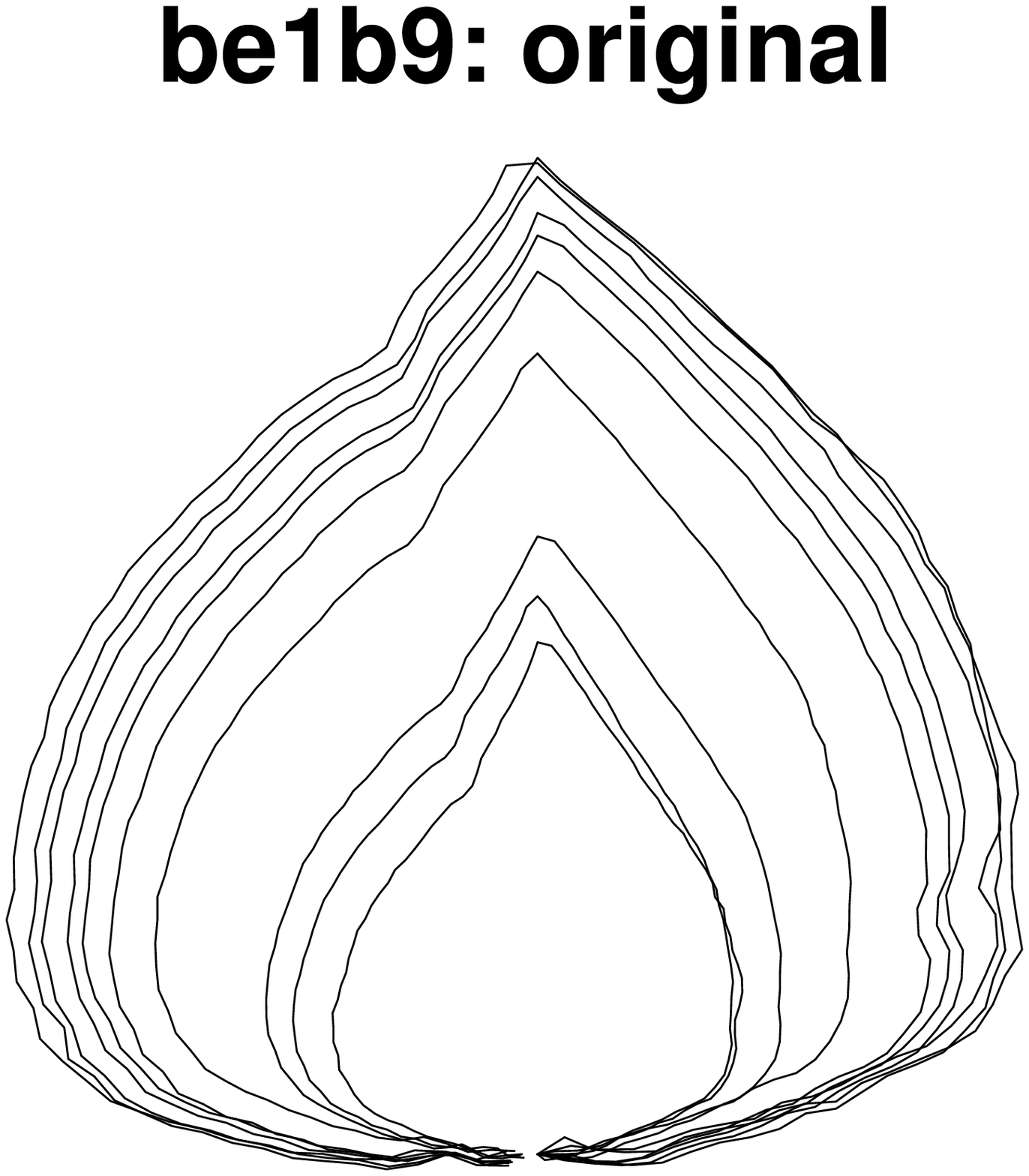}
    	\includegraphics[angle=-90,width=0.45\textwidth]{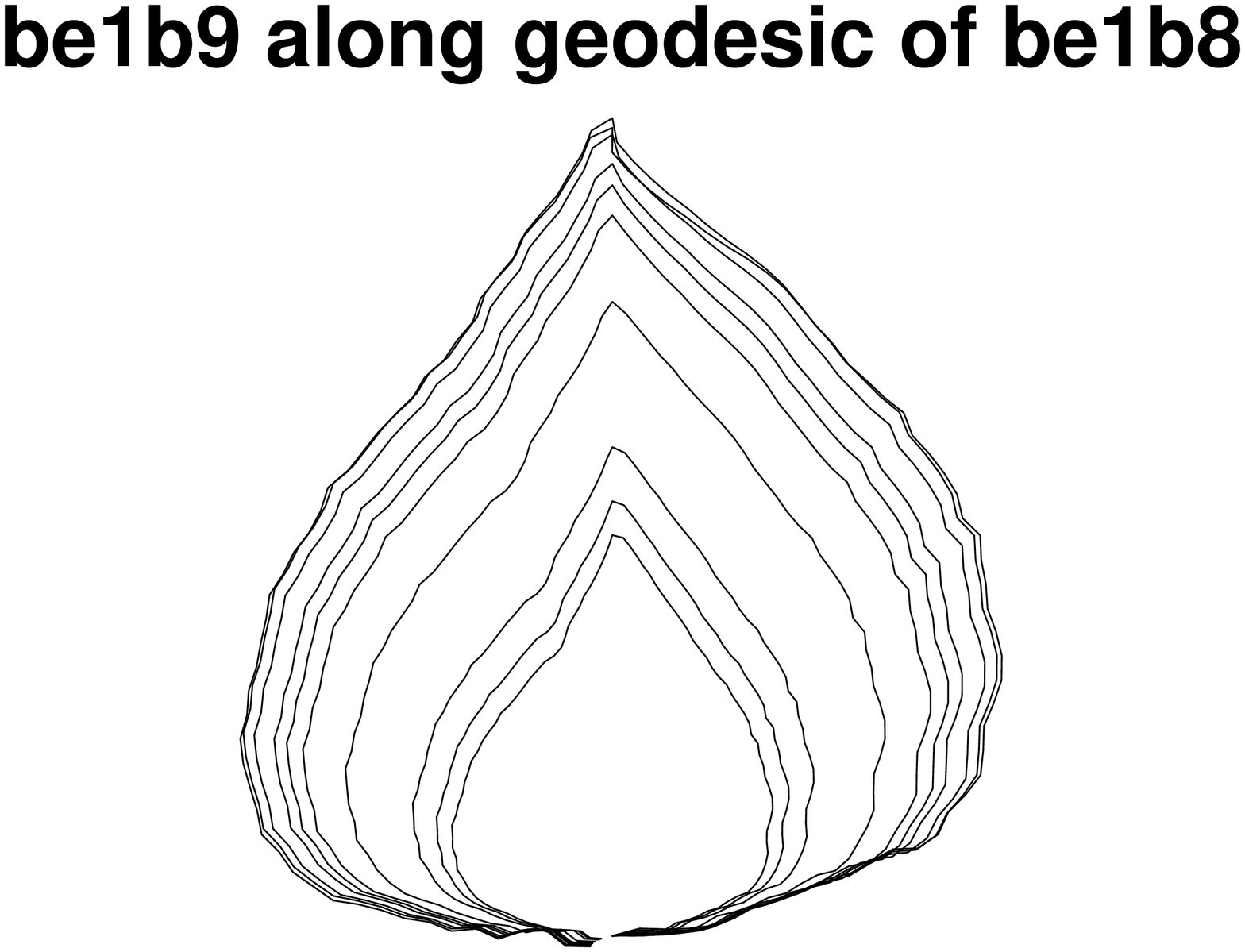}
	 \includegraphics[angle=-90,width=0.45\textwidth]{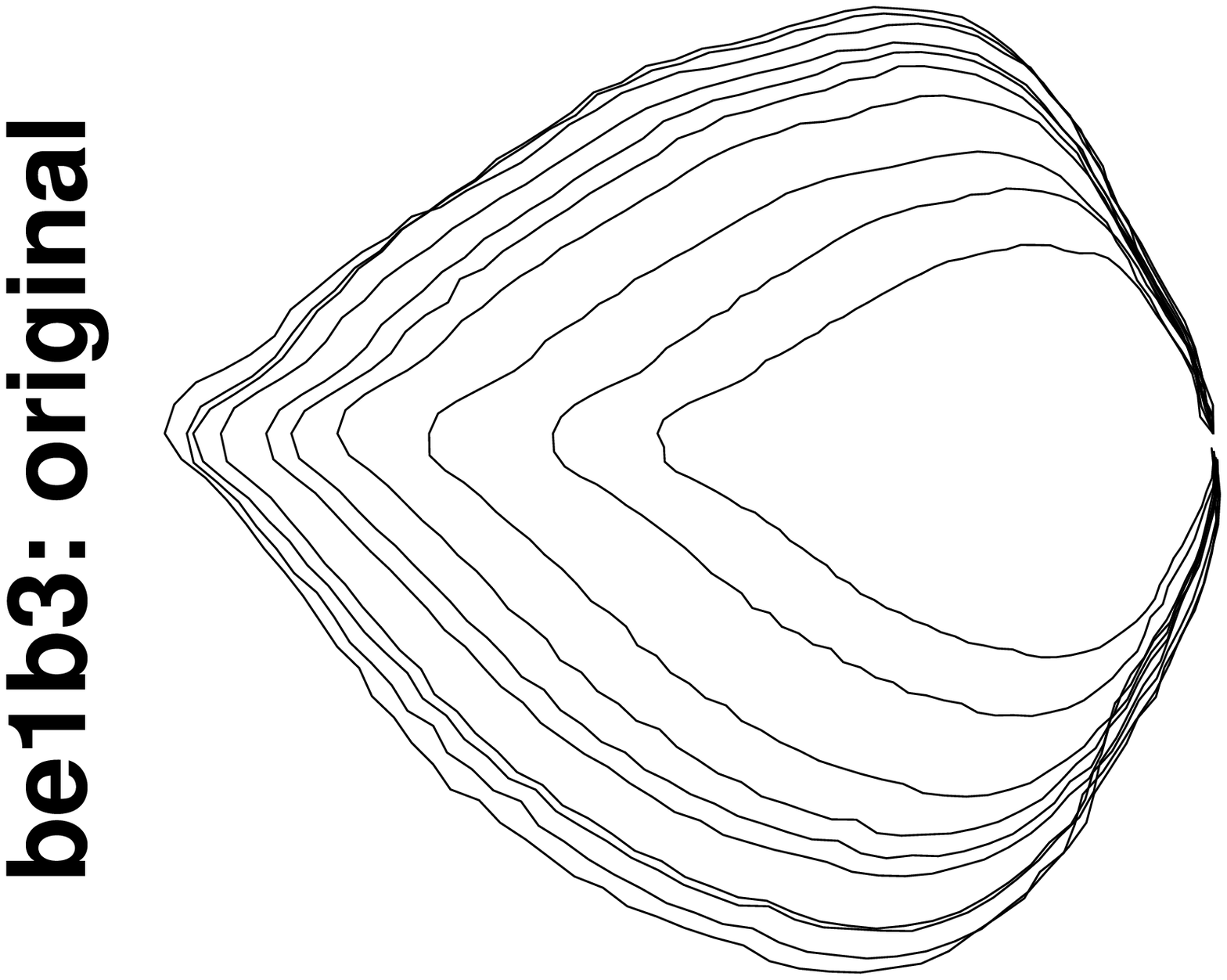}
    	\includegraphics[angle=-90,width=0.45\textwidth]{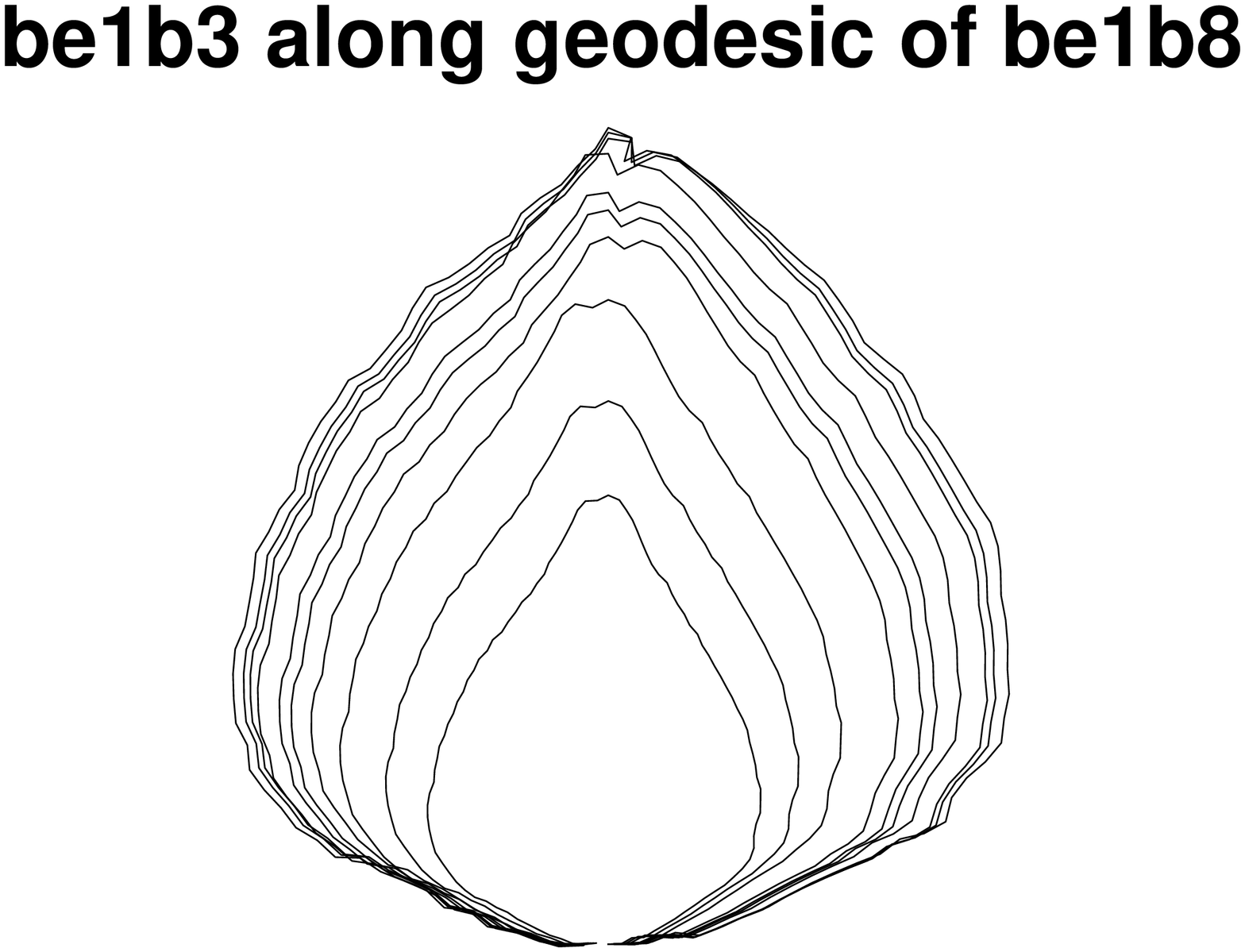}
% 	\caption{\it Shape evolution of black poplar leaves over two weeks. Left column: original contours. Right column: contours obtained from traversing the first contour along the parallel translate of the geodesic from Figure \ref{ori-geod-be1b8.eps}. \label{ori-parallel-tp-be1b8-1.eps}}
% 	\end{figure}
% 	\begin{figure}[h!]
% 	\centering
	 \includegraphics[angle=-90,width=0.45\textwidth]{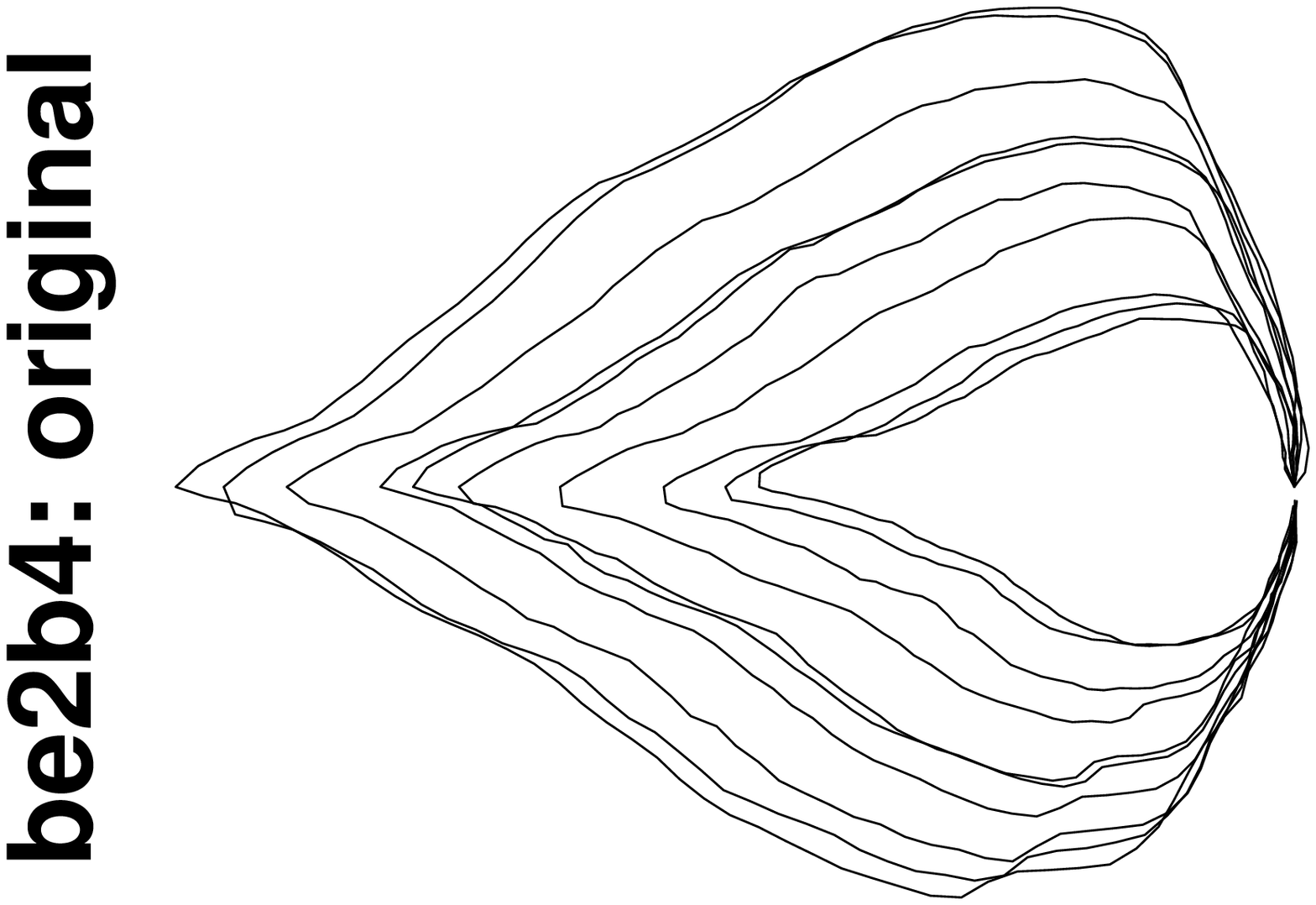}
    	\includegraphics[angle=-90,width=0.45\textwidth]{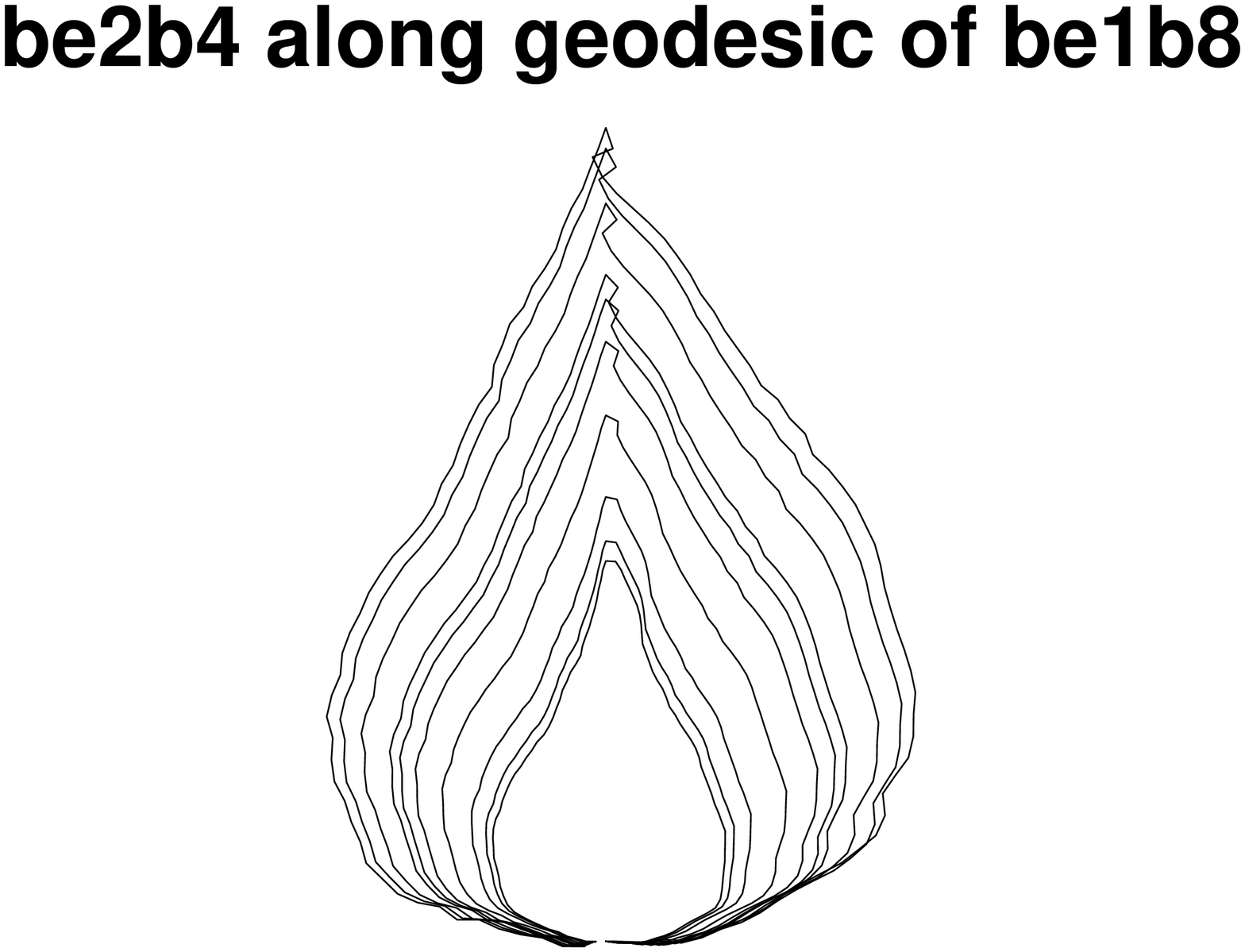}
	 \includegraphics[angle=-90,width=0.45\textwidth]{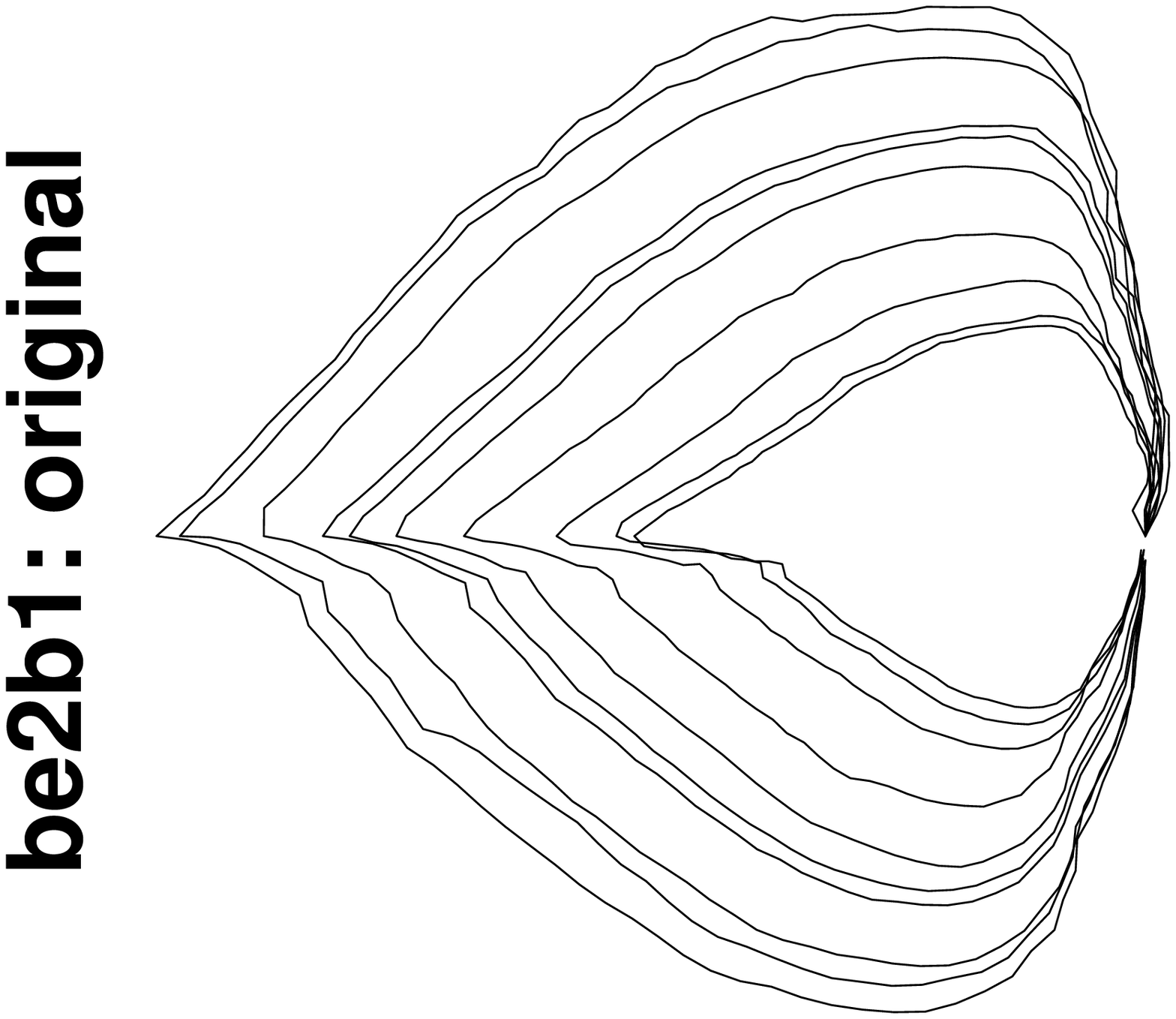}
    	\includegraphics[angle=-90,width=0.45\textwidth]{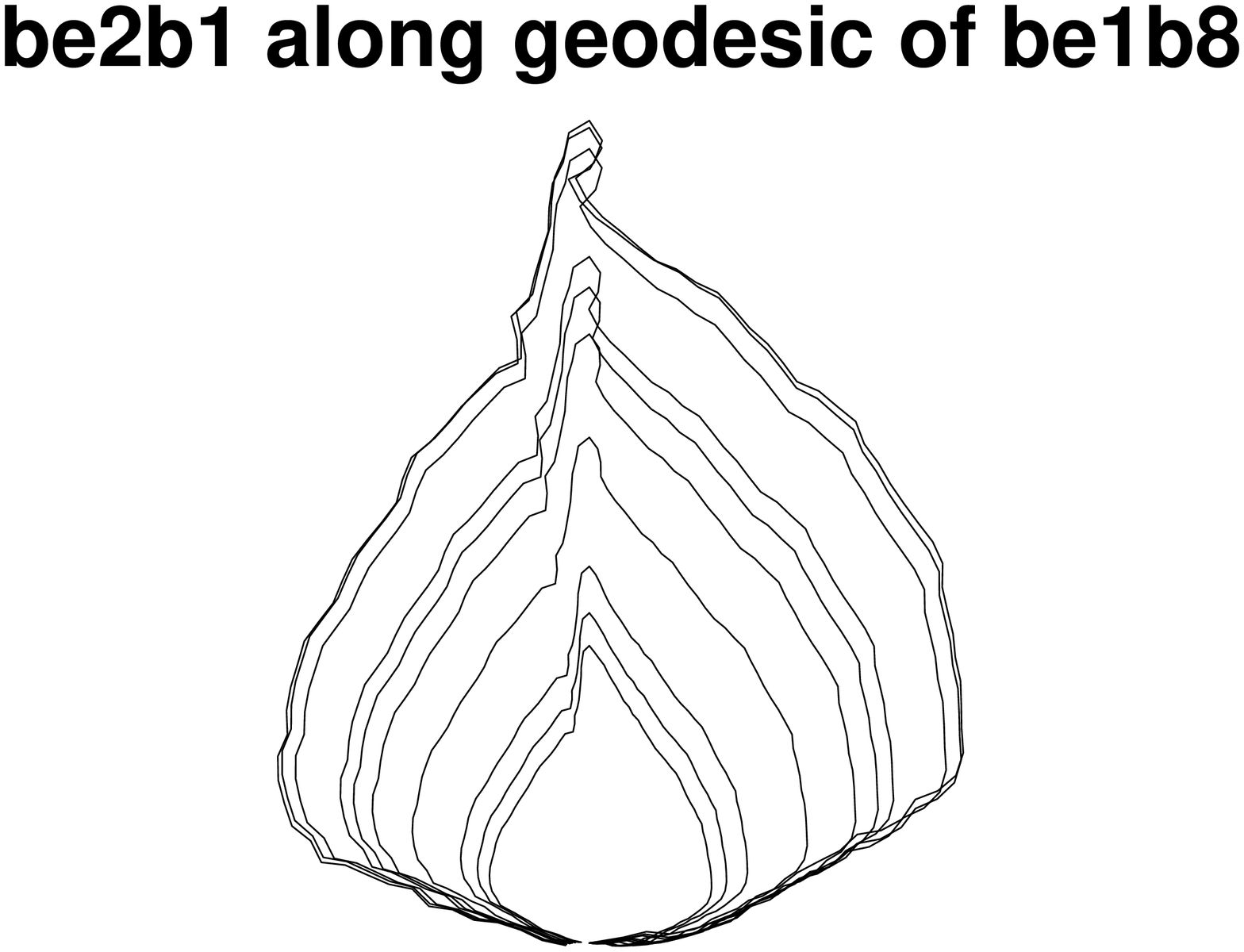}
	\caption{\it Shape evolution of black poplar leaves over two weeks. Left column: original contours. Right column: contours obtained from traversing the first contour along the parallel translate of the geodesic from Figure \ref{ori-geod-be1b8.eps}. \label{ori-parallel-tp-be1b8-2.eps}}
	\end{figure}

	In an application we consider five leaves of a Canadian black poplar tree at an experimental site of the Department of Forest Biometry and Tree Physiology of the University of G\"ottingen. Their contours have been non-destructively extracted over their growing period at 10 approximately evenly spaced days. % over the interval of 14 days. % in July 2008.
	For the following computations we have considered the Euclidean $(2N+1)$-dimensional subspace of $\lw^2$ using Fourier coefficients, $(x_0,x_1,y_1,\ldots,x_N,y_N)$ with $N=100$.
	 
	Leaf ``be1b8'' (left image of Figure \ref{ori-geod-be1b8.eps}) exhibits the most regular shape. In concord with earlier observations %of the temporal evolution 
	of different leaves using landmarks  (cf. \cite{HHGMS07}), the temporal evolution occurs almost along the geodesic determined by initial and terminal shape (as depicted in the right image of Figure \ref{ori-geod-be1b8.eps}). The initial direction of this geodesic has been parallely transplanted to the initial shapes of leaves ``be1b9'', ``be1b3'',	``be2b4'' and ``be2b1''. In the right columns of Figure %s \ref{ori-parallel-tp-be1b8-1.eps} and 
	\ref{ori-parallel-tp-be1b8-2.eps}, the shapes along these new geodesics starting at the corresponding initial shapes have been recorded at the corresponding points in time. The left columns of these figures depict the original temporal shape evolution.

	The common shape dynamics displayed by the original leaf contours (left columns of Figures \ref{ori-geod-be1b8.eps} and \ref{ori-parallel-tp-be1b8-2.eps}) seems two-fold. First, an increase of base angle. Second, different growth ratios are not visible at the apex as its angle remains nearly unchanged. Individual effects are non-symmetric and non uniform lateral growth. Also, leaf 'be1b3' develops a notch left, slightly below the apex, for leaf 'be2b1' an original notch also left, slightly below the apex attenuates.

% 	  (leaf 'be1b8') or not at all apparent (other four leaves). Along the geodesic  the apex More subtly,  
% 	The broad effect of the geodesic: more lateral than distal growth, more subtly, the opposite effect at the very top more acuminating the apex.  more acuminating making the more acute angle is becomes more unchanged getting more acute form nearly unchanged, below apex concave form becomes more convex, lateral growth in reality may be non-symmetric, the geodesic levels that effect different 
% 	in consequ
% 	Around the apex the geodesic tends to acumseems 

	Obviously (right columns of Figures \ref{ori-geod-be1b8.eps} and \ref{ori-parallel-tp-be1b8-2.eps}) leaves  ``be1b9'', ``be1b3' and `be2b4'' follow rather closely the parallel transplant of the geodesic of leaf ``be1b8''. Original non-uniform growth is uniformized and, stronger than originally, apexes acuminate. Even though all of their initial and terminal shapes are quite different, one can say that their temporal evolution is rather similar. This seems to be less the case for leaf ``be2b1''. Its observed growth tends to eliminate its initial strong dent at north-west-north while along the transplanted geodesic, this dent remains, causing increased distal growth at the tip. One could argue that in order to restore an original contour defect, natural growth deviates from its ``original'' plan. Certainly, such phenomena deserve future research. 

	\begin{table}[h!]
	\centering
	\fbox{\begin{tabular}{r|cccc}b&be1b8&be1b3&be2b4&be2b1\\\hline
 	$\rho(v_{\rm b},w_{\rm b})$&$ 0.17  $&$ 0.12  $&$ 0.44$&$0.083 $\\
 	$\mu(v_{\rm b},w_{\rm b})$&$0.99 $&$  0.96  $&$1.0$&$0.88 $
% 	$\mu(v_{\rm b},w_{\rm b})$&$0.98 $&$  0.91  $&$1.0$&$0.80 $
	\end{tabular}}
	\caption{\it Measuring parallelity of geodesics at first shape. Top row: non-central correlation (\ref{non-central-corr:def}) of Fourier coefficients of $v_{\rm b}$ (initial velocity of geodesic approximating leaf shapes of leaf 'b' at its first shape) and parallel transplant $w_{\rm b}$ of $v_{\rm beb18}$ (initial velocity of geodesic approximating leaf shapes of leaf 'be1b8' at its first shape)  to the first shape of leaf 'b'. Second row: the same for the $(1-p)$-values obtained from (\ref{parallel-meas:def}).\label{cor:tab}}
	\end{table}

	As a measure for parallelity, the cosine of the angle between initial velocity of geodesics or equivalently the  the correlation of the respective Fourier coefficients could be taken, cf. first row in Table \ref{cor:tab}. At first glance, in contrast to Figure \ref{ori-parallel-tp-be1b8-2.eps}, these numbers suggest an almost non-existent amount of parallelity. 
	On closer inspection, taking into account, however, that the vectors compared are high-dimensional (of dimension $2N+1=201$), the first and third number of the  second row in Table \ref{cor:tab} indicate high correlation as expressed in the third row: if two random vectors $v,w$ would be independently sampled from a uniform distribution on the $(n-1)$-dimensional unit-sphere, then the density of their angle $\phi$ is proportional to the surface of the $(n-2)$-hypersphere with radius $\sin\phi$ determined by this angle. Based on this consideration we propose the following measure for parallelity of random $v,w\in \mathbb R^n$
	\begin{eqnarray}\label{parallel-meas:def}
%	\begin{array}{rcl} 
	\mu(v,w) &=& 1-\frac{\int_0^{\sqrt{\arccos\rho(v,w)}}\sin^{n-2}\phi\,d\phi} {\int_0^{\pi}\sin^{n-2}\phi\,d\phi}\,,~~\mbox{with }\\\label{non-central-corr:def}
	\rho(v,w) &=& \frac{|\langle v,w\rangle|}{\|v\|\,\|w\|}\,. %,~~\mbox{the non-central correlation.}
%	\end{array}
	\end{eqnarray}

	Finally, let us note that the curves of the contours of the leaves be2b4 and be2b1 projected along the transplanted geodesic of leaf be1b8 start to self-intersect as the leaves grow older. While the effect is rather small, it may indicate that the natural metric of $\Sigma_{ZR}$ be adjusted in order to maintain the hypothesis of geodesic growth.

\section{Discussion and Outlook}

	In this exposition a method to compare shape dynamics has been proposed based on parallel transport of geodesics. While there is quite a few work available, modelling temporal shape evolution by specific curves and splines in shape space ((e.g. \cite{JK87,KMMA01,KDL07}), to the knowledge of the author this is the first time that dynamical aspects of different shapes have undergone a comparison based on the intrinsic geometry of shape space. 

	 In application to Botany, growth of leaves of different shapes has been compared. As underlying shape representation the space $\Sigma_{ZR}$ of closed contours based on angular direction  has been employed. Within this space, in contrast to other models, entire leaf contours can be retrieved in a non-parametric way. In a simple toy example, parallel transport on $\Sigma_{ZR}$ seemed locally similar to parallel transport on Kendall's shape spaces. 

	 This similarity can be rephrased as saying that the landmarks have been ``correctly chosen''. Parallel transport thus may serve as a tool to address an open problem in landmark based shape analysis: optimize number and location of landmark placement for a specific problem at hand. Obviously, too few and wrongly placed landmarks have low predictive power, while too many landmarks reduce power due to undesired variation. Dealing with this latter effect usually requires further methodology, e.g. statistical regularization. 

	In the context of \emph{dynamical shape analysis} for leaves, under optimal landmark placement, since computation of geodesics, e.g. geodesic PCA, is computationally much faster on Kendall's shape spaces than on $\Sigma_{ZR}$, one can perform parallel transport on Kendall's shape spaces and obtain the complete bounding contour in a non-parametric way by mapping to $\Sigma_{ZR}$. This lays out a path %Within this setting, in order to discriminate and classify leaf growth and determine its driving factors, 
	%large scale statistical research seems within reach, thus laying a path 
	for the future steps of the challenging endeveavor of statistically comparing shape dynamics laid out in the Introduction.

% 	It seems like a natural first step for future research to explore under which conditions parallel transport on  $\Sigma_{ZR}$ can be approximated by parallel transport on Kendall's shape spaces. Given the validity of such an approximation, since computation of geodesics, e.g. geodesic PCA, is computationally much simpler on Kendall's shape spaces than on $\Sigma_{ZR}$, in the second step, one can perform parallel transport on Kendall's shape spaces and obtain the complete bounding contour in a non-parametric way by mapping to $\Sigma_{ZR}$. Within this setting, in order to discriminate and classify leaf growth and determine its driving factors, larger scale statistical research seems within reach.

% 	As a second application, intrinsic two-way MANOVA, as introduced in \cite{HHM09}), becomes available for general shape spaces.
% 
% 	In a third application, the geometry of a specific shape space models can be studied more closely by comparing shape deformations of different shapes. Usually decisions by practioners for a specific shape space model have been based based on mathematical grounds such as practicability or mathematical naturality, rather than on insight into practical effects of shape similarity due to the chosen geometry. Comparing parallel transport may aid a practioner to justify a choice of a specific shape space model directly from the application at hand.

\bibliographystyle{../../BIB/elsart-harv}
\bibliography{../../BIB/shape,../../BIB/botany}

\begin{thebibliography}{37}
\expandafter\ifx\csname natexlab\endcsname\relax\def\natexlab#1{#1}\fi
\expandafter\ifx\csname url\endcsname\relax
  \def\url#1{\texttt{#1}}\fi
\expandafter\ifx\csname urlprefix\endcsname\relax\def\urlprefix{URL }\fi

\bibitem[{Blum and Nagel(1978)}]{BN78}
Blum, H., Nagel, R.~N., 1978. Shape description using weighted symmetric axis
  features. Pattern Recognition 10~(3), 167--180.

\bibitem[{Burton(2004)}]{Burton04}
Burton, R.~F., 2004. The mathematical treatment of leaf venation: The variation
  in secondary vein length along the midrib. Ann. of Botany 93~(2), 149--156.

\bibitem[{Dickinson et~al.(1987)Dickinson, Parker, and Strauss}]{DPS87}
Dickinson, T.~A., Parker, W.~H., Strauss, R.~E., 1987. Another approach to leaf
  shape comparisons. Taxon 36~(1), 1--20.

\bibitem[{Dryden and Mardia(1998)}]{DM98}
Dryden, I.~L., Mardia, K.~V., 1998. Statistical Shape Analysis. Wiley,
  Chichester.

\bibitem[{Endress et~al.(2000)Endress, Baas, and Gregory}]{EBG00}
Endress, P.~K., Baas, P., Gregory, M., 2000. Systematic plant morphology and
  anatomy: 50 years of progress. Taxon 49~(3), 401--434.

\bibitem[{Fu and Chi(2006)}]{FuChi06}
Fu, H., Chi, Z., 2006. Combined thresholding and neural network approach for
  vein pattern extraction from leaf images. IEE Proceedings - Vision, Image,
  and Signal Processing 153~(6), 881--892.

\bibitem[{Gielis(2003)}]{Gielis03}
Gielis, J., 2003. A generic geometric transformation that unifies a wide range
  of natural and abstract shapes. invited special paper. American Journal of
  Botany 90, 333--338.

\bibitem[{Gurevitch(1992)}]{Gur92}
Gurevitch, J., 1992. Sources of variation in leaf shape among two populations
  of achillea lanulosa. Genetics 130~(2), 385--394.

\bibitem[{Hearn(2009)}]{Hearn09}
Hearn, D.~J., 2009. Shape analysis for the automated identification of plants
  from images of leaves. Taxon 58~(3), 934--954.

\bibitem[{Hotz et~al.(2010)Hotz, Huckemann, Gaffrey, Munk, and
  Sloboda}]{HHGMS07}
Hotz, T., Huckemann, S., Gaffrey, D., Munk, A., Sloboda, B., 2010. Shape spaces
  for pre-alingend star-shaped objects in studying the growth of plants.
  Journal of the Royal Statistical Society, Series C 59~(1), 127--143.

\bibitem[{Huckemann et~al.(2009)Huckemann, Hotz, and Munk}]{HHM09}
Huckemann, S., Hotz, T., Munk, A., 2009. Intrinsic {MANOVA} for {R}iemannian
  manifolds with an application to {K}endall's space of planar shapes. IEEE
  Transactions on Pattern Analysis and Machine Intelligence.~To appear.

\bibitem[{Huckemann and Ziezold(2006)}]{HZ06}
Huckemann, S., Ziezold, H., 2006. Principal component analysis for {R}iemannian
  manifolds with an application to triangular shape spaces. Adv. Appl. Prob.
  (SGSA) 38~(2), 299--319.

\bibitem[{Jensen(1990)}]{Jensen90}
Jensen, R.~J., 1990. Detecting shape variation in oak leaf morphology: A
  comparison of rotational-fit methods. American Journal of Botany 77~(10),
  1279--1293.

\bibitem[{Jensen et~al.(2002)Jensen, Ciofani, and Miramontes}]{JCM2002}
Jensen, R.~J., Ciofani, K.~M., Miramontes, L.~C., 2002. Lines, outlines, and
  landmarks: Morphometric analyses of leaves of acer rubrum, acer saccharinum
  (aceraceae) and their hybrid. Taxon 51~(3), 475--492.

\bibitem[{Jorgensen and Mauricio(2005)}]{JM05}
Jorgensen, S., Mauricio, R., 2005. Hybridization as a source of evolutionary
  novelty: leaf shape in a hawaiian composite. Genetica 123~(1-2).

\bibitem[{Jupp and Kent(1987)}]{JK87}
Jupp, P.~E., Kent, J.~T., 1987. Fitting smooth path to spherical data. Appl.
  Statist. 36~(1), 34--46.

\bibitem[{Kendall et~al.(1999)Kendall, Barden, Carne, and Le}]{KBCL99}
Kendall, D.~G., Barden, D., Carne, T.~K., Le, H., 1999. Shape and Shape Theory.
  Wiley, Chichester.

\bibitem[{Kent et~al.(2001)Kent, Mardia, Morris, and Aykroyd}]{KMMA01}
Kent, J.~T., Mardia, K.~V., Morris, R.~J., Aykroyd, R.~G., 2001. Functional
  models of growth for landmark data. In: Mardia, K.~V., Aykroyd, R.~G. (Eds.),
  Proceedings in Functional and Spatial Data Analysis. Leeds University Press,
  pp. 109--115.

\bibitem[{Klassen et~al.(2004)Klassen, Srivastava, Mio, and Joshi}]{KSMJ04}
Klassen, E., Srivastava, A., Mio, W., Joshi, S., Mar. 2004. Analysis on planar
  shapes using geodesic paths on shape spaces. IEEE Transactions on Pattern
  Analysis and Machine Intelligence 26~(3), 372--383.

\bibitem[{Krieger et~al.(2007)Krieger, Guralnick, and Smith}]{KGS07}
Krieger, J.~D., Guralnick, R.~P., Smith, D.~M., 2007. Generating empirically
  determined, continuous measures of leaf shape for paleoclimate
  reconstruction. Palaios 22, 212--219.

\bibitem[{Kume et~al.(2007)Kume, Dryden, and Le}]{KDL07}
Kume, A., Dryden, I., Le, H., 2007. Shape space smoothing splines for planar
  landmark data. Biometrika 94~(3), 513--528.

\bibitem[{Lang(1999)}]{L99}
Lang, S., 1999. Fundamentals of Differential Geometry. Springer.

\bibitem[{Le(2003)}]{L03}
Le, H., 2003. Unrolling shape curves. Journal of the London Mathematical
  Society 68~(2), 511--526.

\bibitem[{Le and Kume(2000)}]{LK00}
Le, H., Kume, A., 2000. Detection of shape changes in biological features.
  Journal of Microscopy 200~(2), 140--147.

\bibitem[{Lohmann(1983)}]{Loh83}
Lohmann, G.~P., 1983. Eigenshape analysis of microfossils: A general
  morphometric procedure for describing changes in shape. Mathematical Geology
  15~(6).

\bibitem[{Lu et~al.(2009)Lu, Zhao, and Guo}]{LZG09}
Lu, S., Zhao, C., Guo, X., 2009. Venation skeleton-based modeling plant leaf
  wilting. International Journal of Computer Games Technology 2009, 8.

\bibitem[{Miller et~al.(2006)Miller, Trouv\'{e}, and Younes}]{MTY06}
Miller, M.~I., Trouv\'{e}, A., Younes, L., 2006. Geodesic shooting for
  computational anatomy. J. Math. Imaging Vis. 24~(2), 209--228.

\bibitem[{Morris et~al.(2000)Morris, Kent, Mardia, and Aykroyd}]{MKMA00}
Morris, R., Kent, J.~T., Mardia, K.~V., Aykroyd, R.~G., 2000. A parallel growth
  model for shape. In: Arridge, S., Todd-Pokropek, A. (Eds.), Proceedings in
  Medical Imaging Understanding and Analysis. Bristol: BMVA, pp. 171--174.

\bibitem[{Niinemets et~al.(2007)Niinemets, Portsmuth, and Tobias}]{NPT07}
Niinemets, U., Portsmuth, A., Tobias, M., 2007. Leaf shape and venation pattern
  alter the support investments within leaf lamina in temperate species: a
  neglected source of leaf physiological differentiation? Functional
  Ecology~(1), 28--40.

\bibitem[{O'Neill(1966)}]{O66}
O'Neill, B., 1966. The fundamental equations of a submersion. Michigan Math. J.
  13~(4), 459--469.

\bibitem[{Pizer et~al.(2003)Pizer, Siddiqi, Sz\'{e}kely, Damon, and
  Zucker}]{PSSDZ03}
Pizer, S.~M., Siddiqi, K., Sz\'{e}kely, G., Damon, J.~N., Zucker, S.~W., 2003.
  Multiscale medial loci and their properties. Int. J. Comput. Vision 55~(2-3),
  155--179.

\bibitem[{Ray(1992)}]{Ray92}
Ray, T.~S., 1992. Landmark eigenshape analysis: Homologous contours: Leaf shape
  in syngonium (araceae). American Journal of Botany 79~(1), 69--76.

\bibitem[{Rohlf(1986)}]{R86}
Rohlf, F.~J., 1986. Relationships among eigenshape analysis, fourier analysis,
  and analysis of coordinates. Mathematical Geology 18~(8).

\bibitem[{Schmidt et~al.(2006)Schmidt, Clausen, and Cremers}]{SCC06}
Schmidt, F.~R., Clausen, M., Cremers, D., 2006. Shape matching by variational
  computation of geodesics on a manifold. In: Pattern Recognition (Proc. DAGM).
  Vol. 4174 of LNCS. Springer, Berlin, Germany, pp. 142--151.

\bibitem[{Theophrastus(1976)}]{Th76}
Theophrastus, 1976. De Causis Plantarum, Volume I of III. No. 471 in The Loeb
  classical library. William Heinemann Ltd. and Harvard University Press,
  London and Cambridge, MA, with an English translation by Benedict Einarson
  and George K. K. Link.

\bibitem[{Wolfe(1993)}]{Wolfe93}
Wolfe, J., 1993. A method of obtaining climatic parameters from leaf
  assemblages. US Geological Survey Bull. 2040.

\bibitem[{Zahn and Roskies(1972)}]{ZR72}
Zahn, C., Roskies, R., 1972. Fourier descriptors for plane closed curves. IEEE
  Trans. Computers C-21, 269--281.

\end{thebibliography}
%\bibliography{../../BIB/botany}
%\end{paper}
\end{document}